\documentclass[11pt]{article}
\usepackage[utf8]{inputenc}

\usepackage{amsmath, amssymb, amsthm, amsfonts}
\usepackage{mathtools}
\usepackage{thm-restate}

\usepackage[dvipsnames,svgnames,x11names]{xcolor}
\definecolor{ForestGreen}{rgb}{0.1333,0.5451,0.1333}
\usepackage[colorlinks=true,linkcolor=ForestGreen,citecolor=blue]{hyperref}

\usepackage[margin=1in]{geometry}
\usepackage{graphics}
\usepackage{pifont}
\usepackage{tikz}
\usepackage{bbm}
\usepackage[T1]{fontenc}

\usetikzlibrary{arrows.meta}
\usepackage{environ}
\usepackage{framed}
\usepackage{url}
\usepackage[linesnumbered,ruled,vlined]{algorithm2e}
\usepackage[noend]{algpseudocode}
\usepackage[noabbrev,capitalize,nameinlink]{cleveref}
\crefname{equation}{}{}
\usepackage[labelfont=bf]{caption}
\usepackage{cite}
\usepackage{framed}
\usepackage[framemethod=tikz]{mdframed}
\usepackage{appendix}
\usepackage{graphicx}
\usepackage[textsize=tiny]{todonotes}
\usepackage{tcolorbox}
\allowdisplaybreaks[1]

\newcommand\remove[1]{}

\newtheorem{lemma}{Lemma}[section]
\newtheorem{theorem}{Theorem}
\newtheorem*{lemma*}{Lemma}
\newtheorem{corollary}[lemma]{Corollary}
\newtheorem*{corollary*}{Corollary}
\newtheorem{claim}[lemma]{Claim}

\newtheorem{conj}[lemma]{Conjecture}

\theoremstyle{definition}

\newtheorem*{theorem*}{Theorem}
\newtheorem{definition}[lemma]{Definition}

\newtheorem*{rem*}{Remark}

\newcommand{\eps}{\varepsilon}
\newcommand{\R}{\mathbb{R}}

\renewcommand{\O}{\widetilde{O}}

\crefname{algocf}{Algorithm}{Algorithms}
\crefname{claim}{Claim}{Claims}

\renewcommand{\bar}{\overline}
\renewcommand{\hat}{\widehat}

\newcommand{\cA}{\mathcal{A}}
\newcommand{\cT}{\mathcal{T}}

\renewcommand{\bar}{\overline}

\renewcommand{\deg}{\mathsf{deg}}

\newcommand{\cE}{\mathcal{E}}

\newcommand{\OMv}{\mathsf{OMv}}
\newcommand{\uMv}{\mathsf{uMv}}

\begin{document}

\title{On Approximate Fully-Dynamic Matching and \\ Online Matrix-Vector Multiplication}

\author{Yang P. Liu \\ Institute for Advanced Study \\ yangpliu@ias.edu}

\clearpage\maketitle

\begin{abstract}
We study connections between the problem of fully dynamic $(1-\epsilon)$-approximate maximum bipartite matching, and the dual $(1+\epsilon)$-approximate vertex cover problem, with the online matrix-vector ($\OMv$) conjecture which has recently been used in several fine-grained hardness reductions. We prove that there is an \emph{online} algorithm that maintains a $(1+\epsilon)$-approximate vertex cover in amortized $n^{1-c}\epsilon^{-C}$ time for constants $c, C > 0$ for fully dynamic updates \emph{if and only if} the $\OMv$ conjecture is false. Similarly, we prove that there is an online algorithm that maintains a $(1-\epsilon)$-approximate maximum matching in amortized $n^{1-c}\epsilon^{-C}$ time if and only if there is a nontrivial algorithm for another dynamic problem, which we call dynamic approximate $\OMv$, that has seemingly no matching structure. This provides some evidence against achieving amortized sublinear update times for approximate fully dynamic matching and vertex cover.

Leveraging these connections, we obtain faster \emph{algorithms} for approximate fully dynamic matching in both the online and offline settings.
\begin{itemize}
    \item We give a randomized algorithm that with high probability maintains a $(1-\epsilon)$-approximate bipartite matching and $(1+\epsilon)$-approximate vertex cover in fully dynamic graphs, in amortized $O(\epsilon^{-O(1)} \frac{n}{2^{\Omega(\sqrt{\log n})}})$ update time. This improves over the previous fastest runtimes of $O(n/(\log^* n)^{\Omega(1)})$ due to Assadi-Behnezhad-Khanna-Li [STOC 2023], and $O_{\epsilon}(n^{1-\Omega_{\epsilon}(1)})$ due to Bhattacharya-Kiss-Saranurak [FOCS 2023] for small $\epsilon$. Our algorithm leverages fast algorithms for $\mathsf{OMv}$ due to Larsen and Williams [SODA 2017].
    \item We give a randomized offline algorithm for $(1-\epsilon)$-approximate maximum matching with amortized runtime $O(n^{.58}\epsilon^{-O(1)})$ by using fast matrix multiplication, significantly improving over the runtimes achieved via online algorithms mentioned above. This mirrors the situation with $\mathsf{OMv}$, where an offline algorithm exactly corresponds to fast matrix multiplication. We also give an offline algorithm that maintains a $(1+\epsilon)$-approximate vertex cover in amortized $O(n^{.723}\epsilon^{-O(1)})$ time.
\end{itemize}
\end{abstract}

\newpage

\setcounter{tocdepth}{2}
\tableofcontents

\normalsize
\pagebreak

\section{Introduction}
\label{sec:intro}

Let $G = (V = L \cup R, E)$ denote a bipartite graph undergoing edge updates in the form of edge insertions and deletions. Recently, the problem of maintaining a $(1-\eps)$-approximate bipartite matching of $G$ has seen significant attention. For this dynamic problem, we do not have a good understanding of the optimal complexity. In the more restrictive partially dynamic setting where the graph undergoes only edge insertions (incremental) or only edge deletions (decremental), algorithms with $\O(\eps^{-O(1)})$\footnote{Throughout, we use $\O(\cdot), \tilde{\Omega}(\cdot)$ to hide polylogarithmic factors in $n$ and $\eps^{-1}$. We use $O_{\eps}(\cdot), \Omega_{\eps}(\cdot)$ to denote that the implied constants depend on $\eps$.} amortized update time are known \cite{Gupta14,JJST22,BKS23lp}, even for weighted generalizations. Additionally, under popular hardness assumptions such as the combinatorial Boolean matrix multiplication conjecture or online matrix-vector ($\OMv$) conjecture, it is known that exact versions, even for partially dynamic unweighted graphs, require amortized $\Omega(n^{1-o(1)})$ time.

Thus, recently attention has turned to the complexity of $(1-\eps)$-approximate matching
in the \emph{fully dynamic} setting, where the graph can undergo both edge insertions and deletions. In this setting, a ``dream'' algorithm would maintain a matching in amortized $O((\eps^{-1} \log n)^{O(1)})$ time per update. The existence of such an algorithm has been raised as an open question previously \cite{ARW17}. However, to date it is not even known whether an algorithm with amortized time $O_{\eps}(n^{1-c})$ exists for some absolute constant $c < 1$ independent of $\eps$. The best known bounds are \cite{GP13}, with update time $O(\eps^{-2}\sqrt{m})$ for $m$-edge graphs. For the case of dense graphs (which has been seen to be the difficult case), this runtime is $O(\eps^{-2}n)$. Recently, using Szemeredi's regularity lemma, \cite{ABKL23} gave an algorithm with amortized update time $O(n/(\log^* n)^c)$ for some $c > 0$ for $\eps = (\log^* n)^{-c}$. However, it is unlikely that regularity-based methods can save more than $2^{O(\sqrt{\log n})}$ factors due to Behrend's construction \cite{B46}. To the author's knowledge, these remain the best known bounds for maintaining the matching as $\eps \to 0$.
Very recently, an algorithm that maintained the \emph{size} only of the matching in amortized $O_{\eps}(m^{1/2-\Omega_{\eps}(1)}) = O_{\eps}(n^{1-\Omega_{\eps}(1)})$ time was given by \cite{BKS23b}.

Given the lack of upper bounds, it makes sense to ask whether one can prove a lower bound. Because of the difficulty of proving unconditional lower bounds, there is a trend of proving lower bounds by reduction to certain popular hardness assumptions. A prominent assumption is the online matrix-vector $(\OMv)$-conjecture, introduced by \cite{HKNS15}. This conjecture considers the following problem: an algorithm is given a Boolean matrix $M \in \{0,1\}^{n \times n}$, which is preprocessed in polynomial time. It is then given Boolean vectors $v_1, \dots, v_n$ one at a time. Upon receiving $v_i$, the algorithm must output $Mv_i$ before receiving $v_{i+1}$. The conjecture is that any algorithm which correctly outputs $Mv_i$ for all $i = 1, \dots, n$ must take at least $n^{3-o(1)}$ time. Note that if all vectors $v_1, \dots, v_n$ were given at the same time, then the algorithm could compute all $Mv_i$ using fast matrix multiplication in $O(n^{\omega+o(1)})$ time. There is also a natural combinatorial interpretation of the $\OMv$ problem. To translate, let $M$ be the adjacency matrix of a bipartite graph $G = (V = L \cup R, E)$, which the algorithm receives and preprocesses. Then, $n$ subsets $B_1, \dots, B_n \subseteq R$ are given in order. Upon receiving $B_i$, the algorithm must return which vertices in $L$ have at least one neighbor in $B_i$. Several lower bounds for dynamic problems have been proven assuming the $\OMv$ conjecture, including for \emph{exact} dynamic matching, even in partially dynamic settings.

In this paper, we study the connection between $\OMv$-like phenomena and the $(1-\eps)$-\emph{approximate} fully dynamic matching problem, which to our knowledge, has not been explicitly looked at previously. We also study the dual $(1+\eps)$-approximate vertex cover problem. Towards this, we show that the desirable amortized runtime of $n^{1-c}\eps^{-C}$ time for dynamic matching against adaptive adversaries is achievable if and only if there is an efficient algorithm for a dynamic, approximate version of $\OMv$ that we introduce (see \cref{def:dynapproxomv}, \cref{thm:equivmatch}). Similarly, we also show that there is a fully dynamic algorithm that maintains a $(1+\eps)$-approximate vertex cover in amortized $n^{1-c}\eps^{-O(1)}$ time against adaptive adversaries if and only if the $\OMv$ conjecture itself is false (\cref{thm:equivvtx}).
In both cases, we are able to start with a problem that has matching structure (i.e., dynamic matching/vertex cover) and show equivalence to a problem that no such structure.

Leveraging ideas behind these equivalences, we also obtain new upper bounds for maintaining dynamic matchings and vertex covers. By leveraging the fast $\OMv$ algorithm of \cite{LW17}, we give a randomized $(1-\eps)$-approximate algorithm for dynamic matching, and $(1+\eps)$-approximate dynamic vertex cover, with amortized $O(\eps^{-O(1)}\frac{n}{2^{\Omega(\sqrt{\log n})}})$ update time.
We also obtain faster \emph{offline} randomized algorithms for approximate fully dynamic matching and vertex cover running in amortized time $O(n^{.58}\eps^{-O(1)})$ and $O(n^{.723}\eps^{-O(1)})$ respectively, leveraging fast matrix multiplication. All these algorithms can maintain the edges in the matching and vertices in the vertex cover.

Throughout this paper, when we refer to the dynamic matching or vertex cover problems, we require that the algorithm actually returns the matching/vertex cover, not just its size. Anytime our algorithm only maintains the size of the matching, we will be explicit about it. Several recent works only achieve a size approximation. It is also worth noting that if one can maintain a fractional matching, then there are algorithms that round to integral matchings, even against a fully dynamic \emph{adaptive} adversary (see \cite{Wajc20,BKSW23round} and references therein).

\subsection{Reductions removing matching structure}

To state our reductions precisely, we introduce a few problems. The first is the $\OMv$ (online matrix-vector) problem, introduced in \cite{HKNS15}.
\begin{definition}[$\OMv$ problem]
\label{def:omv}
In the $\OMv$ problem, an algorithm is given a Boolean matrix $M \in \{0,1\}^{n \times n}$. After preprocessing, the algorithm receives an online sequence of query vectors $v^{(1)}, \dots, v^{(n)} \in \{0,1\}^n$. After receiving $v^{(i)}$, the algorithm must respond the vector $Mv^{(i)}$.
\end{definition}
The following is known as the (randomized) $\OMv$ conjecture.
\begin{conj}
\label{conj:omv}
Any randomized algorithm solving the $\OMv$ problem with high probability requires at least $n^{3-o(1)}$ time across preprocessing and queries.
\end{conj}
We show an equivalence between the $\OMv$ conjecture and algorithms for dynamic vertex cover with certain parameters.
\begin{theorem}
\label{thm:equivvtx}
The $\OMv$ conjecture (\cref{conj:omv}) is true if and only if there is \emph{no} randomized algorithm that maintains a $(1+\eps)$-approximate vertex cover in amortized $n^{1-c}\eps^{-C}$ time for any constants $c, C > 0$.
\end{theorem}
Thus, under the $\OMv$ conjecture, the offline runtime we achieve in \cref{thm:mainuppervtx} is not achievable online. This is perhaps unsurprising, since our offline algorithms makes heavy use of fast matrix multiplication. To state our reductions for dynamic matching, we introduce a few more problems in the spirit of $\OMv$ that, to our knowledge, have not been introduced before.
\begin{definition}[Approximate $\OMv$]
\label{def:approxomv}
In the $(1-\gamma)$-approximate $\OMv$ problem, an algorithm is given a Boolean matrix $M \in \{0,1\}^{n \times n}$. After preprocessing, the algorithm receives an online sequence of query vectors $v^{(1)}, \dots, v^{(n)} \in \{0,1\}^n$. After receiving $v^{(i)}$, the algorithm must respond with a vector $w^{(i)} \in \{0,1\}^n$ such that $d(Mv^{(i)}, w^{(i)}) \le \gamma n$, where $Mv^{(i)}$ is the Boolean matrix product, and $d(\cdot, \cdot)$ is the Hamming distance.
\end{definition}
We also introduce a more challenging dynamic version.
\begin{definition}[Dynamic approximate $\OMv$]
\label{def:dynapproxomv}
In the $(1-\gamma)$-approximate dynamic $\OMv$ problem, an algorithm is given a matrix $M \in \{0, 1\}^{n \times n}$, initially $0$. Then, it responds to the following:
\begin{itemize}
    \item $\textsc{Update}(i, j, b)$: set $M_{ij} = b$.
    \item $\textsc{Query}(v)$: output a vector $w \in \{0,1\}^n$ with $d(Mv, w) \le \gamma n$.
\end{itemize}
\end{definition}
We show the following equivalence between dynamic matching and dynamic approximate $\OMv$.
\begin{theorem}
\label{thm:equivmatch}
There is an algorithm solving dynamic $(1-\gamma)$-approximate $\OMv$ with $\gamma = n^{-\delta}$ with amortized $n^{1-\delta}$ for \textsc{Update}, and $n^{2-\delta}$ time for \textsc{Query}, for some $\delta > 0$ against adaptive adversaries, if and only if there is a randomized algorithm that maintains a $(1-\eps)$-approximate dynamic matching with amortized time $n^{1-c}\eps^{-C}$, for some $c, C > 0$ against adaptive adversaries.
\end{theorem}
This implies a reduction from approximate $\OMv$ (\cref{def:approxomv}) to dynamic matching, though we do not know of an equivalence. We do not know of a subcubic time algorithm for approximate $\OMv$, for some $\gamma = n^{-\delta}$.
\begin{corollary}
\label{cor:reduce}
If there is an $(1-\eps)$-approximate dynamic matching algorithm against adaptive adversaries with amortized runtime $n^{1-c} \eps^{-C}$ for constants $c, C > 0$, then the approximate $\OMv$ problem with $\gamma = n^{-\delta}$ can be solved in $n^{3-\delta}$ time for some $\delta > 0$.
\end{corollary}
\begin{proof}
Follows because any algorithm for dynamic approximate $\OMv$ can solve approximate $\OMv$ by calling \textsc{Update} at most $n^2$ times to initialize $M$, which uses total time $n^{3-\delta}$, and then calling \textsc{Query} $n$ times, with total time $n^{3-\delta}$.
\end{proof}

\paragraph{Interpretation of results.} In the author's opinion, the results stated in this section demonstrate the following point. In the maximum matching and vertex cover problems, the challenge in solving the problems efficiently (with $n^{1-c}\eps^{-C}$ update times) does not stem from the matching structure at all. Instead, these problems are difficult in the same way that $\OMv$ is: it is (conjecturally) hard to find edges in induced subgraphs of a graph $G$ in an online manner.

Why might solving approximate $\OMv$ with $\gamma = n^{-\delta}$ actually require almost cubic time? To start, our results use a simple reduction (\cref{lemma:induced}) which shows that a fully dynamic matching algorithm with amortized time $n^{1-c}\eps^{-C}$ can be used to implement the following oracle in \emph{subquadratic} time. Given a bipartite graph $G = (V = L \cup R, E)$ and subsets $A \subseteq L, B \subseteq R$, find a matching on the induced subgraph $G[A, B]$ whose size is within $O(\eps n)$ of optimal. Then, a potentially difficult case is when $G[A, B]$ is itself an induced matching, and the algorithm must locate many edges of this induced matching in subquadratic time. This is very similar to the why designing fast algorithms for $\OMv$ is difficult. Additionally, there are graphs (called \emph{Ruzsa-Szemeredi graphs}) where $G[A, B]$ are induced matchings for many different pairs $(A, B)$.

\subsection{Online dynamic matching and vertex cover}
Leveraging the reductions of \cref{thm:equivvtx,thm:equivmatch} and a fast $\OMv$ algorithm of \cite{LW17}, we give faster algorithms for approximate dynamic matching and vertex cover against adaptive adversaries.
\begin{theorem}
\label{thm:onlineupperm}
There is a randomized algorithm that maintains a $(1-\eps)$-approximate maximum matching on a dynamic graph $G$ in amortized $O(\eps^{-O(1)}\frac{n}{2^{\Omega(\sqrt{\log n})}})$ time against adaptive adversaries.
\end{theorem}

\begin{theorem}
\label{thm:onlineupperv}
There is a randomized algorithm that maintains a $(1+\eps)$-approximate minimum vertex cover on a dynamic graph $G$ in amortized $O(\eps^{-O(1)}\frac{n}{2^{\Omega(\sqrt{\log n})}})$ time against adaptive adversaries.
\end{theorem}

The previous best amortized runtimes (for matching) were a combination of $O(n/(\log^* n)^c)$ of \cite{ABKL23}, via the Szemeredi regularity lemma, or $O_{\eps}(n^{1-\Omega_{\eps}(1)})$ (only maintaining the size), of \cite{BKS23b}. As far as we know, \cref{thm:onlineupperm} is the first sublinear time algorithm for $\eps = (\log n)^{-O(1)}$.
In some sense, our algorithm uses very little graph structure compared to these works. It uses a multiplicative weights method to show equivalence of matching and a dynamic form of $\OMv$, which we then solve in a black-box manner using \cite{LW17}.

\subsection{Offline dynamic matching and vertex cover}

We give a faster offline algorithm for dynamic matching.
\begin{theorem}
\label{thm:mainupper}
There is a randomized algorithm that given an offline sequence of edge insertions and deletions to an $n$-vertex bipartite graph, maintains the edges of a $(1-\eps)$-approximate matching in amortized $O(n^{.58}\eps^{-O(1)})$ time with high probability.
\end{theorem}
Previously, the best known \emph{online} algorithm only maintains the \emph{size} of the matching runs in amortized time $O_{\eps}(m^{1/2-\Omega_{\eps}(1)}) = O_{\eps}(n^{1-\Omega_{\eps}(1)})$, where the dependence is exponential in $\eps^{-1}$ \cite{BKS23b}. In fact, several recent works (based on sublinear matching) only maintain the approximate size of matchings \cite{B23,BRR23,BKS23a,BKSW23}. If the matrix multiplication constant $\omega = 2$, our algorithm in \cref{thm:mainupper} would run in time $\O(\eps^{-O(1)}n^{0.5+o(1)})$.

We also achieve an offline algorithm for approximate vertex cover.
\begin{theorem}
\label{thm:mainuppervtx}
There is a randomized algorithm that given an offline sequence of edge insertions and deletions to an $n$-vertex bipartite graph, maintains a $(1+\eps)$-approximate vertex cover in amortized $O(n^{.723}\eps^{-O(1)})$ time.
\end{theorem}
Once again, the algorithm is able to maintain the exact set of vertices in the vertex cover.

\subsection{Previous work}
\label{subsec:previous}

\paragraph{Fully dynamic matching and vertex cover.} For the problem of maintaining $(1-\eps)$-approximate matchings in fully dynamic graphs, until recently, the best known runtime was $O_{\eps}(\sqrt{m} \eps^{-2})$ \cite{GP13,PS16}. Recently, the runtime was improved to $O(m^{1/2-\Omega_{\eps}(1)})$ \cite{BKS23b}; however, the algorithm can only maintain the \emph{size} of the matching, and not its edges. The \cite{GP13} algorithm can also be extended to maintain $(1-\eps)$-approximate maximum weighted matchings, with polylogarithmic dependence on the maximum weight. However, the dependence on $\eps$ becomes exponential. A black-box reduction from weighted to unweighted matching was shown in \cite{BDL21} (once again, with exponential dependence on $\eps$).

There are also several works studying fully dynamic matching and vertex cover with larger approximation factors \cite{OR10,BHI18,BHN16,BS16,BCH17,RSW22,BK22,GSSU22,Kiss23}. For example, it is known how to maintain maximal matchings (and hence $1/2$-approximate maximum matchings) in amortized polylogarithmic (even constant) time \cite{BGS18,S16}. Recent works have improved this to a $(2-\sqrt{2})$-approximation in polylogarithmic time \cite{B23,BKSW23}. However, the algorithms can only maintain the \emph{size} of the matching again.

Finally, there are algorithms that maintain an exact matching size in fully dynamic graphs in subquadratic $O(n^{1.407})$ time \cite{S07,BNS19}. These algorithms are all based on fast matrix multiplication, to the author's knowledge. Recently, \cite{BC24} studied fully dynamic vertex cover and matching in certain geometric graphs.

\paragraph{Matching in other settings.} There has been significant recent work on partially dynamic matching, where the graph is either incremental (only undergoes edge insertions) or decremental (only undergoes edge deletions). In these settings, $(1-\eps)$-approximation algorithms with amortized runtime $\O(\eps^{-O(1)})$ are known \cite{BLSZ14,D16,Gupta14,BGS20,ABD22}, and recent works have even reduced the dependence on $m$ (hidden by the $\O(\cdot)$) \cite{GLSSS19,BK23esa}, and $\eps$ \cite{BKS23lp}. It is worth mentioning the works \cite{JJST22,BKS23lp}, which maintain a fractional matching incrementally/decrementally by using optimization methods such as multiplicative weights or entropy-regularized optimal transport. Similarly, our algorithms and reductions are based on multiplicative weights algorithms for matching and vertex cover.

Several dynamic matching algorithms are based on \emph{sublinear} matching algorithms.
It is known via standard reductions (see \cref{lemma:contract}) that a $(1-\eps)$-approximate sublinear matching size estimator in time $T(n)$ implies a $(1-\eps-\delta)$-approximate dynamic matching algorithm with amortized time $O_{\eps,\delta}(T(n)/n)$. \cite{BKS23b} in fact designs faster $(1-\eps)$-approximate sublinear matching algorithms, and \cite{BKSW23,B23} are based on $2$-approximate sublinear matching algorithms with $\O(n)$ runtime \cite{B21}. It should be mentioned that there is a recent lower bound showing that any algorithm which estimates the matching size to within $(2/3+\eps)$ requires at least $n^{1.2-o(1)}$ queries \cite{RSW22}. This implies that the approach of using sublinear matching to give $(1-\eps)$-approximate dynamic matching cannot go below $\O(n^{0.2})$ time per update.

\paragraph{Fine-grained complexity.} There is some evidence that improving over $O(\sqrt{m})$ for maintaining an approximate matching in fully dynamic graphs may be difficult. Indeed, there are conditional lower bounds based on the $\mathsf{3SUM}$ and $\OMv$ conjectures for fully dynamic matching algorithms that maintain a matching without length $O(1)$-augmenting paths \cite{AW14,HKNS15,KPP16}. Having no short augmenting paths is a natural way to ensure that the algorithm maintains a $(1-\eps)$-approximate matching. However, the algorithm in \cite{GP13} does not satisfy this property, and more generally, it is unlikely that an algorithm based on ideas from optimization (such as ours) would satisfy this either. Also, it is known that maintaining a matching \emph{exactly} even in incremental or decremental graphs requires $\Omega(n^{1-o(1)})$ amortized time under the $\OMv$ conjecture.

Recently, the NFA acceptance hypothesis was proposed in \cite{BGKL24}. The conjecture is that for any two Boolean matrices $M_0, M_1$, sequence $(b_1, \dots, b_n) \in \{0, 1\}^n$, and $v \in \{0, 1\}^n$, computing the Boolean product $M_{b_1}M_{b_2} \dots M_{b_n} v$ requires at least $n^{3-o(1)}$ time. If true, it immediately implies the $\OMv$ conjecture.

\subsection{Preliminaries}
\label{sec:prelim}

\paragraph{General notation.} We let $[m] = \{1, 2, \dots, m\}$. Let $\vec{0}$ and $\vec{1}$ denote the all-zero and all-ones vectors respectively. We let $T(a, b, c)$ be the runtime needed to multiply a $a \times b$ by $b \times c$ matrix. For vectors $a, b \in \R^n$ we let $a \circ b \in \R^n$ be the vector with entries $(a \circ b)_i := a_ib_i$.
Similarly, we let $a^{-1}$ denote the entry-wise inverse of a vector. We denote $S' \sim_p S$ to denote that $S'$ is a random subset of $S$, where each element in $S$ is included in $S'$ independently with probability $p$.

\paragraph{Graphs.} We let $G = (V = L \cup R, E)$ denote a bipartite graph. Let $\mu(G)$ denote the maximum matching size of $G$. For $A \subseteq L, B \subseteq R$, we let $G[A, B]$ be the induced graph with vertex set $A \cup B$ and edges $\{ (a, b) : a \in A, b \in B, (a, b) \in E(G)\}$. For a graph $H$, we let $\deg^H(v)$ be the degree of $v$ in $H$. For a vertex $v \in V$, let $N(v)$ be the neighbors of $v$, and $N_E(v)$ be the set of neighboring edges to $v$.
We say that a matching $M$ of $G$ is a $(c,\delta n)$-approximate matching of $G$ if $|M| \ge c \cdot \mu(G) - \delta n$. We abbreviate $(c, 0)$-approximate to $c$-approximate for brevity sometimes.

For a bipartite graph $G = (V, E)$ we say that a vector $x \in \R_{\ge0}^E$ is a \emph{fractional matching} if for all $v \in V$, $\sum_{e = (u, v) \in E} x_e \le 1$. In this case, the value of the fractional matching is $\langle \vec{1}, x \rangle = \sum_{e \in E} x_e$. If a graph has a fractional matching of value $\nu$, it also has an integral matching with $\lceil \nu \rceil$ edges. Similarly, a fractional vertex cover $\phi \in \R_{\ge0}^V$ satisfies that $\phi(u) + \phi(v) \ge 1$ for all $(u, v) \in E$, and has value $\langle \vec{1}, \phi \rangle$. If there is a fractional vertex cover with value at most $\nu$, there is a integral vertex cover with value at most $\lfloor \nu \rfloor$.

\subsection{Overview of matching algorithms}
\label{overview:matching}

\subsubsection{Multiplicative weights framework}
\label{subsubsec:mwumatch}

Maximum matching is an instance of a packing linear program. Let us consider the decision version: for a constant $c$, determine whether there is a vector $x \in \R_{\ge0}^E$ satisfying $\sum_{e \in E} x_e = 1$ and $\sum_{e \in N_E(v)} x_e \le c$ for all $v \in V$. The minimal value of $c$ where this is feasible is $c = 1/\mu(G)$, which is achieved by setting $x^*_e = 1/\mu(G)$ for $e \in M$ for some maximum matching $M$ of size $\mu(G)$, and $x^*_e = 0$ for $e \notin M$. It suffices to solve this decision version, by trying all values $c = (1+\eps)^{-i}$ for $0 \le i \le O(\eps^{-1}\log n)$. For simplicity in this overview, let us set $c = 1/\mu(G)$.

We review the analysis of a multiplicative weight update (MWU) algorithm for this packing linear program. The algorithm will use $\O(\eps^{-O(1)})$ iterations. We then discuss how to implement this algorithm in two settings: offline, and using an approximate, dynamic $\OMv$ algorithm (\cref{def:dynapproxomv}). Set $\lambda = \delta \mu(G)$, where $\delta = (\eps/\log n)^{O(1)}$. For $x \in \R^n$ let $f_v(x) := \sum_{e \in N_E(v)} x_e$, and consider the potential $\Phi(x) := \sum_{v \in V} \exp(\lambda f_v(x)).$ In one iteration of the MWU, the algorithm wishes to find a vector $\Delta \in \R_{\ge0}^E$ satisfying the following:
\begin{enumerate}
    \item \label{item:1} $\sum_{e \in E} \Delta_e = 1$.
    \item \label{item:2} $\sum_{e = (u,v) \in E} (\exp(\lambda f_u(x)) + \exp(\lambda f_v(x))) \Delta_e \le (1+\eps)/c \cdot \Phi(x)$.
    \item \label{item:3} For all $v \in V$, $\sum_{e \in N_E(v)} \Delta_e \le \eps/\lambda$.
\end{enumerate}
The final condition says that the \emph{width} of the solution $\Delta$ is small. Given such a $\Delta$, the algorithm updates $x \gets x+\Delta$. After $T$ iterations, the algorithm outputs the vector $\frac{1}{T}x$.

Overall, standard MWU analyses show that if there is a matching of size $\mu(G)$ in $G$, then running the above algorithm for $T = \O(\eps^{-O(1)})$, returns a fractional matching of quality at least $(1-\eps)\mu(G)$ after scaling. In our algorithms, $\Delta$ will be $n$-sparse during each iteration, and thus the fractional matching has at most $O(nT)$ nonzeros, and thus can be rounded to an integral matching efficiently.

It remains to discuss why such a sparse $\Delta$ exists. In fact, one can take $\Delta$ be exactly be the vector $x^*$, the indicator vector on the maximum matching, scaled down by $\mu(G)$. In the following two sections, we discuss how to implement a single iteration of this MWU in their respective settings.

\subsubsection{Offline matching}
By a nice reduction by \cite{Kiss23}, we may assume that $\mu(G) = \Theta(n)$ (see \cref{lemma:contract}). We describe a general strategy for finding $\Delta$ satisfying the above three properties. Let
\[ \hat{E} := \left\{ e = (u,v) \in E : \exp(\lambda f_u(x)) + \exp(\lambda f_v(x)) \le (1+\eps)/c \cdot \Phi(x) \right\}. \]
Clearly, if $\Delta$ is supported on $\hat{E}$, then \cref{item:2} is true. Additionally, note that at most $(1-\eps)|M|$ edges in a matching $M$ violate \cref{item:2} by Markov's inequality. Thus, $\hat{E}$ contains at least $\eps|M|$ edges in $M$, and thus contains a matching of that size. We will ultimately choose $\Delta$ to be supported on a matching of size at least $\Omega(\eps|M|)$, scaled so that $\sum_{e \in E} \Delta_e = 1$. Then, it is clear that \cref{item:3} is also satisfied.

\paragraph{Reducing to maximal matching on induced subgraphs.} It remains to construct such a $\Delta$, which we do formally in \cref{lemma:inducedtomoracle}. For this, we will essentially find a maximal matching supported on $\hat{E}$. Partition the vertices in $L = L_1 \cup \dots \cup L_t$, and $R = R_1 \cup \dots \cup R_t$, where vertices within the same $L_i$ or $R_i$ have values $f_v(x)$ differing by additive $\eps/\lambda$, so that $\exp(\lambda f_v(x))$ differ by multiplicative $(1+\eps)$. One can show that $t \le \O(\eps^{-O(1)})$. We will build a maximal matching by iterating over subgraphs $G[L_i, R_j]$. Let the current maximal matching supported on $\hat{E}$ that we are maintaining involve all vertices in $\bar{A} \cup \bar{B}$, and let $A = L \setminus \bar{A}, B = R \setminus \bar{B}$. Thus, it suffices to find a maximal matching on $G[L_i \cap A, R_j \cap B]$.

\paragraph{Random sampling to reduce degrees.} This is done in \cref{lemma:highdegree}. Let $A' = L_i \cap A, B' = R_j \cap B$, so that we wish to find a maximal matching on $G[A', B']$. We go through the vertices $a \in A'$ one by one. We randomly sample $\O(n/D)$ vertices $b \in B'$ for a parameter $D$, trying to find an edge. If we find an edge, add it to the maximal matching and update $A', B'$. Otherwise, we conclude that the degree of $a$ is at most $D$, with high probability.

\paragraph{Locating edges using matrix multiplication.} We have reduced to the case where $G[A', B']$ satisfies that all vertices $a \in A'$ have degree at most $D$. We will find all these edges in amortized subquadratic time, done formally in \cref{lemma:lowdegree}. 

Recall that in the offline setting, we are solving many instances of this problem in parallel. More concretely, let $G_1, \dots, G_{\tau}$ be graphs such that $G_i$ is the graph after $(i-1) \cdot \eps n$ updates. Because we are assuming that $\mu(G_i) = \Theta(n)$, it suffices to solve the problem on the $G_i$ only.
Also, note that $G_i$ and $G_1$ differ in at most $O(\tau n)$ edges. Our goal is to locate all the edges of the graphs $G_i[A_i', B_i']$, for $i = 1, \dots, \tau$. To do this, we leverage a trick from \cite{WX20} to subsample $A_i'$ to reduce to $\O(D)$ instances where the degrees in $A'$ are $0$ or $1$, and then use Boolean matrix multiplication to find edges in $G_1[A_i', B_i']$ (see \cref{claim:degreeone}). Because $G_1$ and $G_i$ differ in $O(n\tau)$ edges, we can update our edges sets in $O(n\tau)$ time per $i$. Trading off $\tau, D$ properly gives \cref{thm:mainupper}.

\subsubsection{Reducing matching to dynamic $\OMv$}
\label{subsubsec:mwumatch3}
From the above discussion, the only part that used that the dynamic matching problem was offline was in the last paragraph. In this case, we can similarly reduce to finding a maximal matching in an induced subgraph $G[A', B']$, where $G$ is changing dynamically, and the degrees in $A'$ are at most $D$. In fact, it suffices if the maximal matching size also has additive error $O(\delta n)$, for some $\delta = (\eps/\log n)^C$ for sufficiently large $C$. This would only decrease the overall matching size by $\delta n$ per subgraph $G[A', B']$ processed in the algorithm, of which there are $\O(\eps^{-O(1)})$, because there are at most that many pairs $(L_i, R_j)$. Now, using the same trick from \cite{WX20}, we can subsample vertices in $B'$ to reduce to the case where vertices in $A'$ have degree $0$ or $1$, and repeat this $D$ times. Then, we can call an $(1-\gamma)$-approximate $\OMv$ query (\cref{def:dynapproxomv}) to locate these edges, making $O(\gamma n)$ errors. Since this is repeated $D$ times, the total number of errors is $O(\gamma Dn)$. To conclude, we choose $D$ properly in terms of $\gamma, \eps$.

The reason we must use a dynamic approximate $\OMv$ algorithm, and not just an approximate $\OMv$ algorithm is because the graph $G$ on which we make these queries changed due to the updates to the dynamic matching algorithm. We do not know how to simulate a dynamic approximate $\OMv$ algorithm with approximate $\OMv$. On the other hand, it is true that an analogous dynamic $\OMv$ algorithm (without approximations), and standard $\OMv$, are actually equivalent. We discuss this below in \cref{overview:vertexcover} on vertex cover.

The reduction discussed in this section (the only if direction of \cref{thm:equivmatch}) implies in particular that dynamic $\OMv$ (without errors) with subquadratic query time implies a dynamic matching algorithm with sublinear update time. By applying the $\OMv$ algorithm of \cite{LW17} which has amortized query time $n^2/2^{\Omega(\sqrt{\log n})}$, we deduce \cref{thm:onlineupperm}.

\subsubsection{Reducing dynamic $\OMv$ to matching}

\paragraph{Implementing approximate matching queries on induced subgraphs.} $\OMv$ can be thought of as trying to find edges in induced subgraphs of bipartite graphs. Formally, if one wishes to compute the Boolean product $u^\top Mv$ for vectors $u, v$, this is equivalent to being able to decide given a bipartite graph $G = (L \cup R, E)$, whether the induced subgraph $G[A, B]$ for some $A \subseteq L, B \subseteq R$ has at least one edge ($A, B$ are the set of $1$'s in $u, v$ and $M$ is the adjacency matrix of $G$). Thus, our first step is to argue why dynamic matching allows us to implement approximate matching queries on induced subgraphs of $G$ efficiently. Interestingly, previous works used the reverse direction \cite{M05,BKS23b}, and argue that an algorithm that solves approximate matchings on induced subgraphs efficiently can be boosted to one that gives $(1-\eps)$-approximate matchings in $O_{\eps}(1)$ rounds.

Let $\cA$ be a $(1-\eps)$-approximate dynamic matching algorithm for a dynamic graph $G$. We can implement a $(1, O(\eps n))$-matching algorithm on $G[A, B]$ for sets $A, B$ by making $O(n)$ edge insertions to $G$ in the following way. For each vertex $a \in L \setminus A$, feed to $\cA$ an edge insertion from $a$ to a unique isolated vertex for each $a$, and the same for vertices $b \in R \setminus B$. Let this new graph be $G^+$. It is easy to prove that any $(1-\eps)$-approximate matching on $G^+$ can be converted to a $(1, O(\eps n))$-approximate maximum matching on $G[A, B]$. Formally, we have the following.

\begin{lemma}
\label{lemma:induced}
Let $G = (V = L \cup R, E)$ be a bipartite graph, and $A \subseteq L, B \subseteq R$. Let $G^+$ be the graph $G$ with additional edges $(a, a')$ for each $a \in L \setminus A$, $(b', b)$ for each $b \in R \setminus B$, where $a', b'$ are all distinct. Let $M$ be a $(1-\eps)$-approximate matching on $G^+$. Then $M[A, B]$ is a $(1,2\eps n)$-approximate matching on $G[A, B]$.
\end{lemma}
\begin{proof}
Given a matching $M$ on $G^+$, it can be converted into a matching $M'$ with $|M'| \ge |M|$, and which uses all edges $(a, a')$, $(b', b)$, in $O(n)$ time. Indeed, if $a$ is matched to $\bar{a}$, then replace $(a, \bar{a})$ with $(a, a')$. Thus, $\mu(G^+) = \mu(G[A,B]) + (n-|A|) + (n-|B|)$, and if $M$ is the returned matching, $|M| \ge (1-\eps)\mu(G^+) \ge \mu(G^+) - 2\eps n$, because $\mu(G^+) \le 2n$. Also,
$|M[A, B]| \ge |M| - (n - |A|) - (n - |B|) \ge \mu(G[A, B]) - 2\eps n$, where the first inequality follows because any edge in $M$, but not in $M[A, B]$, must involve at least one vertex in $(L \setminus A) \cup (R \setminus B)$.
\end{proof}

\paragraph{Solving approximate-$\OMv$: } Assume that $\cA$ solves $(1-\eps)$-dynamic matching in amortized $\cT := n^{1-c} \eps^{-C}$ time, for any $\eps$. We will use $\cA$ to give an algorithm for dynamic $O(D\eps)$-approximate-$\OMv$ (\cref{def:dynapproxomv}) with update time $\cT$ and query time $\O(n^2/\sqrt{D} + n\cT)$, for any $D$. For $D = \eps^{-1/2}$, this implies a dynamic $(1-O(\sqrt{\eps}))$-approximate $\OMv$ algorithm with amortized query time $\O(\eps^{1/4} n^2 + n^{2-c}\eps^{-C})$, which is subquadratic for a proper choice of $\eps$.

We now explain the reduction. Pass edge updates directly to $\cA$. Now, let $M$ be the adjacency matrix of bipartite graph $G$, and let the query vector $v$ be the indicator of a set $B \subseteq R$. Computing $Mv$ is equivalent to determining which vertices in $L$ are adjacent to some vertex in $B$, and we are allowed some errors.

\paragraph{Reducing degrees.} Let $w$ be our output vector. We want that $w$ agrees with $Mv$ in many coordinates. To start, pick $O(n^2/D_L \cdot \log n)$ random pairs $(\ell, b) \in L \times B$. If $(\ell, b)$ is an edge, set $w_{\ell} = 1$ permanently. Let $A$ be the set of $\ell \in L$ where $w_{\ell} = 0$ still. Then with high probability, for all $a \in A$, $\deg^{G[L, B]}(a) \le D_L$, i.e., the degree of $a$ to vertices in $B$ is at most $D_L$. The runtime of this step is $\O(n^2/D_L)$.

We now perform a similar procedure to vertices in $B$. We wish to determine approximately which vertices in $b \in B$ have $\deg^{G[A,B]}(b) \ge D_R$. To do this, choose $O(n^2/D_R \cdot \log n)$ random pairs $(a, b) \in A \times B$. If $(a, b)$, then loop over all neighbors $a'$ of $b$, and mark $w_{a'} = 1$ permanently. Note that this step takes $O(n)$ time, and not $O(\deg^{G[A,B]}(b))$ time, because we do not have query access to the adjacency list of $G[A, B]$. Then, remove $b$ from $B$. Because $\deg^{G[L, B]}(a) \le D_L$, out of the $O(n^2/D_R \cdot \log n)$ random pairs, only $\O(nD_L/D_R)$ will be edges. Thus, this step takes $\O(n^2D_L/D_R)$ time. Because any vertex $b \in B$ with $\deg^{G[A,B]}(b) \ge D_R$ is found with high probability, after this phase, we have reduced to solving the problem on $G[A, B]$, which has maximum degree $D_R$.

\paragraph{Peeling off matchings.} Let us set $D_R := D$ and $D_L = \sqrt{D}$, so that $G[A, B]$ has maximum degree $D$, and the previous step took $\O(n^2/\sqrt{D})$ time. Now consider the induced subgraph $G[A, B]$. If it has at most $O(D \eps n)$ edges, we terminate and return $w$. This may cause us to set $w_a = 0$ incorrectly for $O(D \eps n)$ edges. Thus, the output is valid for $O(D\eps)$-approximate-$\OMv$. Otherwise,
$G[A, B]$ has more than $O(D \eps n)$ edges, and hence has a matching of size at least $O(\eps n)$, because the maximum degree of $G[A,B]$ is at most $D$. Use algorithm $\cA$ to find a $(1, O(\eps n))$-approximate matching on $G[A, B]$ using the algorithm of \cref{lemma:induced}.
The matching will have at least $\eps n$ edges, say involving vertices $A' \subseteq A$. Set $w_a = 1$ for all $a \in A'$, and update $A \gets A \setminus A'$. Because $A$ is decreasing, this phase can be implemented by making $O(n)$ edge insertion calls to $\cA$, because we can implement induced subgraph queries with edge insertions as discussed above (\cref{lemma:induced}). Thus, the total time is $O(n \cT)$. The formal argument is in \cref{subsec:ifmatch}.

\subsection{Overview of vertex cover algorithms}
\label{overview:vertexcover}

\subsubsection{Multiplicative weights framework}
There is a dual covering linear program for the vertex cover problem, which we formally cover in \cref{subsec:mwuvertexcover}. Its decision version can be stated as determining whether there is a vector $y \in \R_{\ge0}^V$ with $\sum_{v \in V} y_v = 1$ and $y_u + y_v \ge c$ for all edges $(u, v) \in E$. Note that the maximum possible value of $c$ is $1/\mu(G)$.
Similarly to matching, we can define the potential function
\[ \Phi(y) := \sum_{(u,v) \in E} \exp(-\lambda(y_u+y_v)), \] for $\lambda = (\eps/\log n)^{O(1)}\mu(G).$ In this setting, the potential will decrease over the course of the algorithm. Each iteration, the algorithm will try to find as many vertices $v \in V$ as possible where \[ \sum_{u \in N(v)} \exp(-\lambda(y_u+y_v)) = \exp(-\lambda y_v) \sum_{u \in N(v)} \exp(-\lambda y_u) \] is above some threshold, the analogue of \cref{item:2} above. Formally, see \cref{def:voracle}. Below, we discuss how to solve this subproblem in the various settings we consider.

\subsubsection{Offline vertex cover}
Once again, for simplicity we assume that $\mu(G) = \Theta(n)$. Formalizing this is actually quite challenging, for the following reason. In the case of matching, previous works \cite{Kiss23,BKS23b} showed that random contractions preserved the matching with high probability (\cref{lemma:contract}). Critically, a matching in the contracted graph can be lifted to a matching on the original graph trivially. This is much less evident for vertex cover: if $a, b$ are contracted into the same vertex, which is included in a vertex cover, how do we create a vertex cover on the uncontracted graph? However, we are able to give such a reduction against oblivious adversaries, while the reduction in \cref{lemma:contract} is against adaptive adversaries. It's worth noting that if our only goal was to achieve runtimes of the form $n^{1-c}\eps^{-O(1)}$ for some positive constant $c$, this reduction is not necessary. However, we find the reduction to be interesting, and it also improves the runtime. We discuss further in \cref{subsec:vtxomv}.

Given this, our offline vertex cover algorithm is simple. We can wait $\eps n$ updates between computing the vertex cover, since we are assuming the size is $\Theta(n)$. Now, over multiple subproblems, we can use fast matrix multiplication to evaluate all the quantities $\sum_{u \in N(v)} \exp(-\lambda y_u)$: this is multiplying the adjacency matrix of $G$ by a vector. As before, the graph $G$ changes over time, but over $\tau$ subproblems we wish to solve, the number of edges differs by at most $O(n\tau)$, so we can update the contribution of these edges directly.

\subsubsection{Reducing vertex cover to $\OMv$}
\label{subsec:vtxomv}
We first show how to reduce vertex cover to dynamic $\OMv$, i.e., approximate dynamic $\OMv$ without errors. In the case where $\mu(G) = \Theta(n)$, this is similar to the previous section, except to compute $\sum_{u \in N(v)} \exp(-\lambda y_u)$ we cannot use a real-valued matrix-vector product. Instead, we estimate this sum using Boolean matrix-vector and random sampling. Once again, partition vertices by grouping vertices with similar $y_u$ values. Thus, our problem reduces to estimating $|N(v) \cap S|$ for all $v$ and a fixed subset $S$. If this quantity is at least $D$, then we can estimate it in $O(\eps^{-O(1)}n/D)$ random samples per vertex. Otherwise, we use the dynamic $\OMv$ oracle to locate all the remaining $O(nD)$ edges, similar to as done in the matching case.

It remains to discuss how to reduce dynamic $\OMv$ to static $\OMv$, and why it morally suffices to consider the case $\mu(G) = \Theta(n)$. For the latter, it is known the amortized cost of updates during times when $\mu(G) \approx k$ is $\O(k \eps^{-2})$, from work of \cite{GP13}. Their idea is that one can maintain a dynamic graph with $O(k^2)$ edges whose vertex cover is identical to the graph $G$.
This runtime is sublinear unless $k \approx n$.

To reduce dynamic $\OMv$ to $\OMv$ (which we do formally in \cref{lemma:dynomv}), the idea is to split the graph $G$ into many subgraphs of size $n^{\alpha} \times n^{\alpha}$. We initialize an $\OMv$ algorithm on each of them. When a subgraph undergoes an edge update, we mark it, and in the future we brute force the matrix-vector products on that subgraph in $O(n^{2\alpha})$ time. When a large fraction of the subgraphs undergo an edge update, we reinitialize.

The reduction in this section proves the if direction of \cref{thm:equivvtx}. Once again, by applying the $\OMv$ algorithm of \cite{LW17} with this section, we deduce \cref{thm:onlineupperv}, a faster fully dynamic $(1+\eps)$-approximate vertex cover algorithm.

\subsubsection{Reducing $\OMv$ to vertex cover}
Similar to matching, we can use a dynamic vertex cover algorithm to implement calls to approximate vertex cover on induced subgraphs.
\begin{lemma}
\label{lemma:inducedvertex}
Let $G = (V = L \cup R, E)$ be a bipartite graph, and $A \subseteq L, B \subseteq R$. Let $G^+$ be the graph $G$ with additional edges $(a, a')$ for each $a \in L \setminus A$, $(b', b)$ for each $b \in R \setminus B$, where $a', b'$ are all distinct. Let $C$ be a $(1+\eps)$-approximate vertex cover on $G^+$. Then $C \cap (A \cup B)$ is a $(1, 2\eps n)$-approximate vertex cover on $G[A, B]$.
\end{lemma}
\begin{proof}
We may assume that $(L \setminus A) \subseteq C$ and $(R \setminus B) \subseteq C$. Indeed, consider an edge $(a, a')$ for $a \in L \setminus A$. If $a' \in C$, then we can simply replace it with $a$ (this is only better).
Because $\mu(G^+) = 2n-|A|-|B|+\mu(G[A, B])$, we know that $|C \cap (A \cup B)| = |C| - (2n-|A|-|B|) \le (1+\eps)\mu(G^+) - (2n-|A|-|B|) \le \mu(G[A,B]) + 2\eps n$, as desired.
\end{proof}
Towards using this oracle to solve $\OMv$, recall that there is a subcubic time $\OMv$ algorithm over $n$ queries if and only if there is a subcubic time $\uMv$ algorithm over $n$ queries (as shown in \cite{HKNS15}): given subsets $A, B$, determine whether $G[A, B]$ has at least one edge. To solve this, use \cref{lemma:inducedvertex}. If $G[A, B]$ is empty, then $C \cap (A \cup B)$ has at most $2\eps n$ vertices. We can brute force all edges with at least one endpoint in $C$ in time $O(\eps n^2)$. Setting $\eps = n^{-\delta}$ completes the algorithm.

\paragraph{Outer loops not based on MWU.} There are many outer iterative procedures we could have based our algorithms on. For example, \cite{BKS23b} uses a combinatorial outer loop of \cite{M05} which essentially reduces maximum matching to $\eps^{-O(1/\eps)}$ instances of finding a constant-factor approximate maximum matching an induced subgraph $G[A, B]$ which is guaranteed to have a matching of size at least $\Omega_{\eps}(n)$. We use MWU over this because we wish for our algorithms to have polynomial dependence on $\eps$, instead of exponential. Another optimization algorithm for solving maximum matching, vertex cover, and more generally, optimal transport, is the Sinkhorn algorithm \cite{Cut13,AWR17}. While this algorithm also have polynomial dependence on $\eps$, its iterations require estimating row and column sums of a matrix undergoing row and column rescalings. While this would suffice for our offline algorithms, with the same bounds up to $\eps^{-O(1)}$ factors, it does not interact as nicely with the setup of \cref{thm:equivmatch}. This is because it is difficult to implement an iteration of the Sinkhorn algorithm with a dynamic approximate $\OMv$ oracle: if there is a row in the matrix with a single large entry, we cannot afford to simply ignore it, while we can in the MWU algorithm described above in \cref{subsubsec:mwumatch3}.

\section{Dynamic Bipartite Matching}
\label{sec:matching}

\subsection{Multiplicative weights framework}
\label{subsec:mwumatching}
In this section, we will give a multiplicative weights procedure for finding an approximately maximum fractional matching. We consider the decision version of the maximum matching problem on a bipartite graph $G = (V, E)$: determine whether there is a vector $x \in \R_{\ge0}^E$ satisfying $\sum_{e \in E} x_e = 1$ and $\sum_{e \in N_E(v)} x_e \le c$ for all $v \in V$. Set $f_v(x) := \sum_{e \in N_E(v)} x_e$, a parameter $\lambda$, and define the potential function $\Phi(x) := \sum_{v \in V} \exp(\lambda f_v(x))$. Let $x^{(0)} = 0$ so that $\Phi(x^{(0)}) \le n^2$.

We abstractly define an oracle that given $x$, finds a vector $\Delta \in \R_{\ge0}^E$ to add that does not increase the potential significantly.
\begin{definition}[MWU oracle]
\label{def:moracle}
Given graph $G = (V, E)$, $x \in \R_{\ge0}^E$, and $c, \lambda, \eps > 0$, we say that $\Delta$ is the output of an \emph{MWU oracle} for matching if it satisfies the following:
\begin{enumerate}
    \item \label{item:m1} $\sum_{e \in E} \Delta_e = 1$.
    \item (Value) \label{item:m2} $\sum_{e = (u,v) \in E} \Delta_e(\exp(\lambda f_u(x)) + \exp(\lambda f_v(x))) \le (1+\eps)c \cdot \Phi(x).$
    \item (Width) For all $v \in V$, $\sum_{e \in N_E(v)} \Delta_e \le \eps/\lambda.$ \label{item:m3}
\end{enumerate}
\end{definition}
We show that calling the oracle $T = \O(\eps^{-O(1)})$ times gives a fractional matching with value at least $(1-\eps)/c$. We start by bounding the potential increase in one iteration.
\begin{lemma}
\label{lemma:phiincrease}
Let $\Delta$ be an output of the MWU oracle. Then $\Phi(x+\Delta) \le (1 + (1+3\eps)\lambda c)\Phi(x)$.
\end{lemma}
\begin{proof}
Recall that for $x \le \eps$, $\exp(x) \le 1+x+x^2 \le 1+(1+\eps)x$. Thus,
\begin{align*}
    \Phi(x+\Delta) &= \sum_{v \in V} \exp(\lambda f_v(x))\exp(\lambda f_v(\Delta)) \le \sum_{v \in V} \exp(\lambda f_v(x))(1 + (1+\eps)\lambda f_v(\Delta)) \\
    &\le \Phi(x) + (1+\eps)\lambda \sum f_v(\Delta)\exp(\lambda f_v(x)) \le \Phi(x) + (1+\eps)\lambda \cdot (1+\eps) c \cdot \Phi(x) \\
    &\le \Phi(x) + (1+3\eps)\lambda c \cdot \Phi(x) \le (1 + (1+3\eps)\lambda c)\Phi(x).
\end{align*}
Here, the first inequality uses \cref{def:moracle} \cref{item:m3}, and the second line uses \cref{item:m2}.
\end{proof}
This lets us upper bound the value of $\Phi$ after $T$ iterations.
\begin{corollary}
\label{cor:phifinal}
Let $x^{(0)} = 0$ and $x^{(i+1)} = x^{(i)} + \Delta^{(i)}$ where $\Delta^{(i)}$ is a MWU oracle for $x^{(i)}$. Then $\Phi(x^{(T)}) \le \exp((1+3\eps)T\lambda c)n^2$.
\end{corollary}
From this we can extract a nearly feasible vector for the linear program.
\begin{corollary}
\label{cor:xfinal}
Let $x^{(0)} = 0$ and $x^{(i+1)} = x^{(i)} + \Delta^{(i)}$ where $\Delta^{(i)}$ is a MWU oracle for $x^{(i)}$. Then the vector $\bar{x} := \frac{1}{T}x^{(T)}$ satisfies $\sum_{e \in E} \bar{x}_e = 1$ and $\sum_{u \in N_E(v)} \bar{x}_e \le (1+3\eps)c + \frac{2\log n}{\lambda T}$ for all $v \in V$.
\end{corollary}
\begin{proof}
We know that $\exp(\lambda f_v(x^{(T)})) \le \Phi(x^{(T)}) \le \exp((1+3\eps)T\lambda c)n^2$ by \cref{cor:phifinal}. The result follows by taking logarithms and dividing by $\lambda T$.
\end{proof}

\subsection{Offline dynamic matching}
\label{subsec:offlinematching}

\subsubsection{Contractions}
Recent works have shown that when solving $(1-\eps)$-approximate dynamic matching, we can assume that the graph has a maximum matching of size at least $\Omega(\eps n)$. For this reduction, we need the notion of a \emph{contraction}.
\begin{definition}[Contraction]
\label{def:contraction}
For a function $\phi: V \to V_{\phi}$, let $G_{\phi}$ be the graph with edge set $E_{\phi} = \{(\phi(u), \phi(v)) : (u, v) \in E(G), \phi(u) \neq \phi(v)\}$. We say that $\phi$ is a \emph{contraction} of $G$, with corresponding graph $G_{\phi}$.
\end{definition}
Note that for any contraction $\phi$, $\mu(G_{\phi}) \le \mu(G)$.
\begin{lemma}[\!{\cite[Lemma 7.2]{BKS23b}}]
\label{lemma:contract}
There exists a dynamic algorithm $\cA$ with $\O(1)$ worst case update time, which maintains a set of $K = \O(1)$ contractions $\{\phi_1, \dots, \phi_K\}$ with corresponding graphs $\{G_{\phi_1}, \dots, G_{\phi_K}\}$, and a subset $I \subseteq K$. Throughout the sequence of updates (with high probability against an adaptive adversary) the algorithm ensures that: (i) $|V_{\phi_i}| = \Theta(\mu(G)/\eps)$ for all $i \in I$, and (ii) there is an $i^* \in I$ such that $(1-\eps)\mu(G) \le \mu(G_{\phi_{i^*}}) \le \mu(G)$.
\end{lemma}
Given a matching on a contraction $G_{\phi}$ of $G$, it is straightforward to return the corresponding matching in $G$ of the same size.

\subsubsection{From induced subgraph maximal matching queries to a MWU oracle}
We use maximal matchings queries on induced subgraphs of $G$ to build a MWU oracle.
\begin{definition}
\label{def:ismm}
For a graph $G = (V = L \cup R, E)$ and $A \subseteq L, B \subseteq R$, we say that $M$ is a \emph{$\beta$-near maximal matching} on $G[A, B]$ if $M$ is a matching, and at most $\beta$ edges in $G[A, B]$ are not adjacent to a matched vertex in $M$.
\end{definition}
\begin{lemma}
\label{lemma:inducedtomoracle}
Let $\lambda = \eps^2/(10c)$, $T = 20 \eps^{-3} \log n$, and $c \ge 1/\mu(G)$. There is an algorithm that implements a MWU oracle as in \cref{def:moracle} by making $\O(\eps^{-4})$ of the following form: return a $\beta$-near maximal matching of $G[A, B]$ for $A \subseteq L, B \subseteq R$ for $\beta = \tilde{\Omega}(\eps^5 \mu(G))$, plus $\O(n \eps^{-O(1)})$ additional time.
\end{lemma}
\begin{proof}
Let $x$ be the current point in the MWU. 
By \cref{cor:xfinal}, we know that $\exp(\lambda f_v(x)) \le \exp(2T\lambda c)n^2 \le \exp(O(\eps^{-1}\log n))$. Let $\tau = O(\eps^{-2}\log n)$, and let $L_i = \{u \in L : \exp(\lambda f_u(x)) \in [(1+\eps/10)^{i-1}, (1+\eps/10)^i]$ for $i = 1, 2, \dots, \tau$. Define $R_i$ similarly. Note that $L = L_1 \cup \dots \cup L_{\tau}$ and $R = R_1 \cup \dots \cup R_{\tau}$. Finally, define $\hat{E} = \bigcup\{G[L_i, R_j] : (1+\eps/10)^i + (1+\eps/10)^j \le (1+\eps)c \cdot \Phi(x) \}.$
Let $M$ be a maximal matching on the edges $\hat{E}$, and define $\Delta_e = 1/|M|$ for $e \in M$ and $\Delta_e = 0$ otherwise. We claim that $\Delta_e$ is a valid output to the MWU oracle.

\cref{item:m1} is clear. \cref{item:m2} follows because all edges in $(u, v) \in \hat{E}$ satisfy $\exp(\lambda f_u(x)) + \exp(\lambda f_v(x)) \le (1+\eps)c \cdot \Phi(x)$ by the definition of $L_i, R_j$ and $\hat{E}$. To show \cref{item:m3} it suffices to argue that $1/|M| \le \eps/\lambda = 10c/\eps$, or $|M| \ge \eps/(10c)$, which follows from $|M| \ge \eps\mu(G)/10$ as $c \ge 1/\mu(G)$ by assumption.
Let $M^*$ be a maximum matching in $G$ of size $\mu(G)$, and let $\Delta^*_e = 1/\mu(G)$ for $e \in M^*$ and $0$ otherwise. Then we know that:
\[ \sum_{e \in M^*} \exp(\lambda f_u(x)) + \exp(\lambda f_v(x)) \le \Phi(x). \]
Thus, by Markov's inequality, the number of edges $e \in M^*$ with $\exp(\lambda f_u(x)) + \exp(\lambda f_v(x)) \le (1+\eps/2)c \cdot \Phi(x)$ is at least:
\[ |M^*| - \frac{1}{(1+\eps/2)c} \ge \mu(G) - (1-\eps/3)\mu(G) = \eps/3 \cdot \mu(G), \]
where we have used $c \ge 1/\mu(G)$. All such edges must be in $\hat{E}$ by definition. Thus, every maximal matching of $\hat{E}$ contains at least $\eps \cdot \mu(G)/6$ edges.

Finally, we discuss how to discuss how to find a maximal matching of $\hat{E}$. Maintain the sets of currently unmatched vertices $A \subseteq L, B \subseteq R$. Iterate over pairs $(i, j)$ with $(1+\eps/10)^i + (1+\eps/10)^j \le (1+\eps)c \cdot \Phi(x)$. Now, query for a $\beta$-near maximal matching on $G[A \cap L_i, B \cap R_j]$, add it to the current matching, and update $A, B$ accordingly. Because there are at most $\O(\eps^{-4})$ pairs $(i, j)$, the number of queries is $\O(\eps^{-4})$. We argue that the returned matching $M$ must have at least $\eps \mu(G)/10$ edges. For each $(i, j)$, let $E_{i,j}$ be the set of edges that are potentially not adjacent to an endpoint of $M$, where $|E_{i,j}| \le \beta$. Then the only edges in $\hat{E}$ that are potentially not adjancent to $M$ are in $\bigcup_{i,j} E_{i,j}$. This is at most $\O(\beta \eps^{-4}) \le \eps \mu(G)/20$ edges. Thus, $M$ can be turned into a maximal matching on $\hat{E}$ by adding at most $\mu(G)/20$ edges, so $|M| \ge \eps\mu(G)/6 - \eps\mu(G)/20 \ge \eps\mu(G)/10$.
\end{proof}

\subsubsection{Implementing a MWU oracle offline}
\label{subsec:mwuoracle}
We work with the following general setup. We have graphs $G_1, \dots, G_t$, all with the same vertex set $L \cup R$, on $n$ vertices. Additionally, $\mu(G_i) \ge \eps n$ for all $i \in [t]$. Finally, each $G_i$ differs from $G_1$ in at most $\Gamma$ edges. Our goal is to implement a MWU oracle as in \cref{def:moracle} for all the graphs $G_1, \dots, G_t$ simultaneously. To see why this corresponds to dynamic offline matching, the reader can think of $G_i$ as the graph after $\eps^2 n(i-1)$ updates. We do not have to consider the graphs between $G_i$ and $G_{i+1}$ in the update sequence because we can afford $\eps^2n$ additive error. Then $\Gamma = \eps^2nt$.

Using \cref{lemma:inducedtomoracle}, it suffices to solve the following problem: return a maximal matching on $G_i[A_i, B_i]$, for sets $A_i \subseteq L, B_i \subseteq R$. We start with the low-degree setting.
\begin{lemma}
\label{lemma:lowdegree}
Let $G_1, \dots, G_t$ be graphs such that $G_i$ and $G_1$ differ in at most $\Gamma$ edges. Let $A_i \subseteq L, B_i \subseteq R$ for $i \in [t]$, such that for all $a \in A_i$, $\deg^{G_i[A_i,B_i]}(a) < D$. There is a randomized algorithm that finds all edges in $G_i[A_i, B_i]$ for all $i \in [t]$ with high probability in total time
\[ \O\left(D \cdot T(n/D, n, t) + t\Gamma \right). \]
\end{lemma}
We first show a subclaim that an algorithm can, for every vertex in $A_i$ of degree at most one, find its neighboring edge. This lemma has a few differences in hypotheses from the previous. First, we say that $G_i$ and $G_1$ differ from $\Gamma_i$ edges, depending on $i$. Also, the graphs $G_i$ have vertex set $L \cup R'$, where $|L| = n$, and $|R| = n' \ll n$. This is the setting we will apply the claim in, when we prove \cref{lemma:lowdegree}.
\begin{claim}
\label{claim:degreeone}
Let $G_1, \dots, G_t$ be bipartite graphs with vertex set $L \cup R'$, where $|L| = n$ and $|R'| = n'$. Also, say $G_i$ and $G_1$ differ in at most $\Gamma_i$ edges for $i \in [t]$. Let $A_i \subseteq L, B_i' \subseteq R'$ for $i \in [t]$. There is an algorithm, that for all $i \in [t]$ and $a \in A_i$ with $\deg^{G_i[A_i,B_i']}(a) = 1$, finds the adjacent edge to $a$ in $G_i[A_i,B_i']$, in total time
\[ \O\left(\sum_{i\in[t]} \Gamma_i + T(n, n', t)\right). \]
\end{claim}
\begin{proof}
The algorithm is based on a nice idea of \cite{WX20}. Label the vertices in $R'$ with $\{1, 2, \dots, n'\}$. Construct subsets $R^{(j)} \subseteq R'$ for $j \le \lceil \log_2 n' \rceil$, where $R^{(j)}$ contains all $r \in R'$ whose label has $j$-th bit in binary equals $1$. For any $a \in A_i$, $i \in [t]$, we can determine $\deg^{G_i[A_i,B_i'\cap R^{(j)}]}(a)$ in total time
\[ \O\left(\sum_{i\in[t]} \Gamma_i + T(n, n', t)\right) \] in the following way. First, find $\deg^{G_1[A_i,B_i'\cap R^{(j)}]}(a)$ for all $a \in A_i$ by multiplying the adjacency matrix of $G_1$ (which is $n \times n'$) by the matrix of indicator vectors of $B_i'$ for $i \in [t]$ (which is $n' \times t$). Then, find $\deg^{G_i[A_i,B_i'\cap R^{(j)}]}(a)$ by updating the $\Gamma_i$ edges where $G_1$ and $G_i$ differ. If $\deg^{G[A_i,B_i']}(a) = 1$, and $a$ is adjacent to $b \in B_i'$, then the set of $j$ with $\deg^{G[A_i,B_i'\cap R^{(j)}]}(a) = 1$ is exactly the bits in the label of $b$. Thus, we can recover $b$.
\end{proof}
We can reduce from the degree $D$ to degree one case by subsampling.
\begin{proof}[Proof of \cref{lemma:lowdegree}]
For $j = 1, 2, \dots, \O(D)$, let $R^{(j)} \sim_{1/D} R$ be independently sampled subsets. Note that $|R^{(j)}| \le \O(n/D)$ with high probability. Now, instantiate \cref{claim:degreeone} for the setting $R' = R_j$ (so $n' = |R_j| \le \O(n/D)$) and $B_i' = B_i \cap R_j$.

We argue that every edge in $G_i[A_i, B_i]$ is found with high probability. Fix an edge $(a, b)$ adjacent to $a \in A_i$. For a fixed $j$, $(a, b)$ is the only edge in $G_i[A_i, B_i \cap R^{(j)}]$ with probability at least $1/D \cdot (1-1/D)^{D-1} \ge 1/(3D)$. Thus, with high probability there exists $j \le \O(D)$ where $(a, b)$ is the unique edge adjacent to $a$ in $G_i[A_i, B_i \cap R^{(j)}]$.

Finally we bound the runtime. The contribution of the $T(n,n',t)$ terms in \cref{claim:degreeone} is $\O(D \cdot T(n, n/D, t))$. For a fixed $j$, the expectation of the $\sum_{i\in[t]} \Gamma_i$ terms is at most $t\Gamma/D$, as each edge differing between $G_i$ and $G_1$ survives in $G_i[L, R^{(j)}]$ with probability $1/D$. Thus, the total contribution over $j = 1, \dots, \O(D)$ is $\O(t\Gamma)$ as claimed.
\end{proof}
Now we give an algorithm for general subset maximal matching queries.
\begin{lemma}
\label{lemma:highdegree}
Let $G_1, \dots, G_t$ be graphs such that $G_i$ and $G_1$ differ in at most $\Gamma$ edges. Let $A_i \subseteq L, B_i \subseteq R$ for $i \in [t]$. There is a randomized algorithm that returns a maximal matching on each $G_i[A_i, B_i]$ for $i \in [t]$ with high probability in total time
\[ \O\left(t\Gamma + n^2t/D + D \cdot T(n, n/D, t) \right), \]
for any parameter $D$.
\end{lemma}
\begin{proof}
For each $i \in [t]$, do the following procedure. Iterate over vertices $a \in A_i$. Take $\O(n/D)$ random samples over vertices $b \in B_i$. If $(a, b)$ is an edge, add it to the maximal matching for $G_i[A_i, B_i]$, remove $a, b$ from $A_i, B_i$ respectively, and move onto the next $a \in A_i$. If no edge was found, also continue on to the next $a \in A_i$. The total time over $G_1, \dots, G_t$ is $\O(n^2t/D)$. Also, for each vertex remaining in $A_i$ that was unmatched, its degree within $G_i[A_i, B_i]$ is at most $D$, with high probability, or else an adjacent edge would be found with high probability over $\O(n/D)$ samples. Then, we apply \cref{lemma:lowdegree} to find all remaining edges in $G_i[A_i, B_i]$ to trivially find a maximal matching in linear time. The total runtime follows from \cref{lemma:lowdegree}.
\end{proof}
\cref{thm:mainupper} follows by combining \cref{lemma:highdegree} with \cref{lemma:inducedtomoracle}.
\begin{proof}[Proof of \cref{thm:mainupper}]
By \cref{lemma:contract}, it suffices to consider the case where the graphs in the dynamic update sequence have $s$ vertices, and have maximum matching size at least $\eps s$, for some $s \le n$. Let $t = s^x$, for some $x \in [0, 1]$ chosen later. Let $G_i$ be the graph after $(i-1)\eps^2 s$ updates, for $i \in [t]$.
Thus, $G_i$ and $G_1$ differ in at most $\Gamma \le ts$ edges. By \cref{lemma:inducedtomoracle} and \cref{cor:xfinal}, running a MWU to find a $(1-\eps)$-approximate fractional matching on $G_i$ uses $\O(\eps^{-4}T) = \O(\eps^{-7})$ calls to a maximal matching oracle on induced subgraphs (\cref{def:ismm}). Each call can be implemented for all $G_1, \dots, G_t$ in total time $\O\left(t\Gamma + s^2t/D + D \cdot T(s, s/D, t) \right)$ by \cref{lemma:highdegree}. Let $D = s^y$. Note that this only returns a fractional matching. However, the support is size at most $\O(\eps^{-O(1)}s)$, because each iteration adds a matching to $x$. Thus, it can be rounded to a integral matching in time $\O(\eps^{-O(1)}s)$ which is negligible.
The amortized runtime over $\eps^2 st$ updates is then:
\[ \O\left(\eps^{-O(1)}(t\Gamma + s^2t/D + D \cdot T(s, s/D, t))/(st)\right) = \O\left(\eps^{-O(1)}(s^x + s^{1-y} + s^{y-1-x} \cdot T(s, s^{1-y}, s^x)\right). \]
For the choices $x = 0.579, y = 0.421$ (found use \cite{Complexity}), the amortized runtime is $O(\eps^{-O(1)}n^{.58})$.
\end{proof}

\subsection{Equivalence of dynamic matching and approximate dynamic matrix-vector}
\label{subsec:reducematching}

\subsubsection{Matching from dynamic matrix-vector}
Assume that there is an algorithm for $(1-\gamma)$-approximate dynamic matrix-vector (\cref{def:dynapproxomv}), for $\gamma = n^{-\delta}$, with update time $\cT := n^{1-\delta}$, and query time $n\cT = n^{2-\delta}$. We show how to use this to implement a $\beta$-near maximal matching oracle. The proof closely follows the approach in the previous section: first reduce degrees, and then subsample to reduce to degree one graphs.
\begin{lemma}
\label{lemma:mmoracle}
There is a randomized algorithm that on a graph $G = (V = L \cup R, E)$ with subsets $A \subseteq L, B \subseteq R$, returns a $\O(D\gamma n)$-near maximal matching on $G[A, B]$ in amortized $\O(n^2/D + Dn\cT)$ time for any parameter $D$, with high probability.
\end{lemma}
\begin{proof}
Iterate over $a \in A$, and sample $\O(n/D)$ vertices $b \in B$. If there is an edge, add it to the maximal matching, and remove $a, b$ from $A, B$ respectively. Otherwise, continue. This costs $\O(n^2/D)$ total time, and with high probability we have reduced to the case where $\deg^{G[A,B]}(a) \le D$ for all $a \in A$. We give a procedure to find all edges in $G[A, B]$.

We repeat the argument of \cref{lemma:lowdegree,claim:degreeone}. For $j = 1, 2, \dots, \O(D)$ let $B_j \sim_{1/D} B$ be randomly sampled sets. With high probability, for every edge $(a, b) \in G[A, B]$, we have that $b \in B_j$ and $\deg^{G[A, B_j]}(a) = 1$ for some $j$. Fix such a $j$. Express the labels of vertices in $B_j$ in binary, and let $B_j^{(k)}$ be the vertices whose $k$-th bit is a $1$. If $M$ is the adjacency matrix of $G$, we can compute the Boolean product $M 1_{B_j^{(k)}}$  up to error in $\gamma n$ coordinates, in time $n\cT$ by assumption.
Any vertex $a \in A$ whose answers are never wrong correctly finds its neighboring edge.
Since we query this oracle $\O(D)$ times, the total number of edges in $G[A, B]$ that we do not find is bounded by $\O(D\gamma n)$, and the total runtime is $\O(Dn\cT)$. Finally, we return a maximal matching over the set of edges that the algorithm finds, which by definition is $\O(D\gamma n)$-near to a maximal matching.
\end{proof}

\begin{proof}[Proof of only if direction of \cref{thm:equivmatch}]
Let $\cA$ be a $(1-\gamma)$-approximate dynamic $\OMv$ data structure.
By \cref{lemma:contract}, we may assume that the graphs in our dynamic sequence have $s$ vertices, and maximum matching size at least $\eps s$.
We can wait $\eps^2 s$ updates between graphs on which we must compute an approximate maximum matching. We directly pass these updates to $\cA$ to update the matrix $M$, which will be the adjacency matrix of $G$. This uses amortized $\cT = s^{1-\delta}$ time.

Given a graph $G$ whose adjacency matrix is $M$, 
by \cref{lemma:inducedtomoracle}, calling a $\beta$-near maximal matching oracle a total of $\O(\eps^{-7})$ times suffices to implement an MWU that returns a fractional matching of value at least $(1-\eps)\mu(G)$, with sparsity $O(sT) = \O(s \eps^{-3})$, which can then be rounded to an integral matching. We apply \cref{lemma:mmoracle} with $D = s^{\delta/2}$ to implement each oracle call. For $\delta > 0$, $\O(D\gamma s) = \O(s^{1-\delta/2}) \le \tilde{\Omega}(\eps^5 \mu(G))$, because $\mu(G) \ge \eps s$, and thus \cref{lemma:mmoracle} successfully implements the required oracle of \cref{lemma:inducedtomoracle}. By \cref{lemma:mmoracle}, the time of this step is $\eps^{-O(1)} \cdot \O(s^2/D + Ds\cT) = \eps^{-O(1)} \cdot \O(s^{2-\delta/2})$. Amortized over $\eps^2 s$ steps between matching computations, the amortized time is $\O(\eps^{-O(1)} s^{1-\delta/2})$ as desired.
\end{proof}

We now show \cref{thm:onlineupperm} by combining this reduction with the $\OMv$ algorithm of \cite{LW17}.
\begin{theorem}[\!\!{\cite[Theorem 1.1]{LW17}}]
\label{thm:lw}
There is a randomized algorithm for $\OMv$ against adaptive adversaries with amortized time $n^2/2^{\Omega(\sqrt{\log n})}$.
\end{theorem}
Recall that \cref{thm:equivmatch} shows that matching is equivalent to approximate \emph{dynamic} $\OMv$. Thus, we need to argue that static $\OMv$ implies dynamic $\OMv$.
We start by using an $\OMv$ algorithm to implement a dynamic $\OMv$ algorithm. Let the \emph{dynamic $\OMv$} problem be as in \cref{def:dynapproxomv}, with $\gamma = 0$.
\begin{lemma}
\label{lemma:dynomv}
Assume there is an algorithm that solves the $\OMv$ problem against adaptive queries in total update and query time $n^{3-\delta}$. Then there is an algorithm that solves the dynamic $\OMv$ problem in amortized update time $n^{1-\delta/5}$ and query time $n^{2-\delta/5}$ against adaptive adversaries.
\end{lemma}
\begin{proof}
Partition $[n] = S_1 \cup S_2 \cup \dots \cup S_t$ with $|S_i| = n^{\alpha}$ for all $i \in [t]$, so $t = n^{1-\alpha}$. Initialize an $\OMv$ algorithm $\cA_{i,j}$ for $G[S_i,S_j]$ on all pairs $(i, j) \in [n]^2$. Initializing all these data structures costs total time $n^{2-2\alpha} n^{(3-\delta)\alpha} = n^{2+\alpha-\delta\alpha}$. Let $M$ be the adjancency matrix of $G$.

At each point in time, the algorithm remembers the set of updates since the last rebuild. Say there were $u$ updates and $q$ queries. If $u \ge n^{2-3\alpha}$, or $q \ge n^\alpha$, rebuild the data structure. Otherwise, implement a query on vector $v$ as follows. For $(i, j)$ such that $S_i \times S_j$ contains an update since the last rebuild, compute $M[S_i, S_j]v$ directly in time $n^{2\alpha}$. Otherwise, use the $\OMv$ data structure on $M[S_i, S_j]$ (which hasn't changed), in amortized time $n^{2\alpha(1-\delta)}$ (this is valid since we only make $q \le n^\alpha$ queries before rebuilding).
Thus, the amortized time ignoring rebuilds is at most \[ n^{2-3\alpha}n^{2\alpha} + n^{2-2\alpha} \cdot n^{2\alpha(1-\delta)} \le O(n^{2-2\alpha\delta}). \]
The amortized time of rebuilding because of updates is: $n^{2+\alpha-\delta\alpha} / n^{2-3\alpha} \le n^{4\alpha}$. The amortized cost of rebuilding due to queries is: $n^{2+\alpha-\delta\alpha}/n^{\alpha} = n^{2-\delta\alpha}$. We can take $\alpha = 1/5$ to conclude.
\end{proof}
\begin{corollary}
\label{cor:dynomv}
There is a randomized algorithm for dynamic $\OMv$ against adaptive adversaries with amortized update time $n/2^{\Omega(\sqrt{\log n})}$ and query time $n^2/2^{\Omega(\sqrt{\log n})}$.
\end{corollary}
Combining \cref{cor:dynomv} with the only if direction of \cref{thm:equivmatch} which we have established above shows \cref{thm:onlineupperm}.

\subsubsection{Dynamic matrix-vector from matching}
\label{subsec:ifmatch}
In this section we show the if direction of \cref{thm:equivmatch}, by describing an algorithm which takes a matching algorithm $\cA$ and applies it to give an algorithm for dynamic approximate-$\OMv$ (see \cref{def:dynapproxomv}).
\begin{definition}[Dynamic matching algorithm]
\label{def:matchingalgo}
We say that $\cA$ is a $(1-\eps)$-approximate dynamic matching algorithm with amortized runtime $\cT$ if it supports the following operations on a dynamic bipartite graph $G$:
\begin{itemize}
    \item $\cA.\textsc{Insert}(e), \cA.\textsc{Delete}(e)$: Insert/delete edge $e$ to/from $G$, in amortized time $\cT$,
    \item $\cA.\textsc{Matching}()$: Returns a $(1-\eps)$-approximate maximum matching $M$, with high probability against an adaptive adversary, in time $O(|M|)$.
\end{itemize}
\end{definition}
In \cref{algo:omv}, $\textsc{ApproxOMv}$ takes bipartite graph $G$, and $B \subseteq R$, corresponding to the matrix $M$ and vector $v$ respectively.
\begin{algorithm}[!ht]
  \caption{Approximate-$\OMv$ via Dynamic Matching \label{algo:omv}}
  \SetKwProg{Globals}{global variables}{}{}
  \SetKwProg{Proc}{procedure}{}{}
  \Globals{}{
    $n$: number of vertices in $G$, \\
    $\cA$: $(1-\eps)$-dynamic matching algorithm with amortized time $\cT$ (\cref{def:matchingalgo}), \\
    $D$: degree parameter, $D_L \gets \sqrt{D}$, $D_R \gets D$.
  }
  \Proc{$\textsc{ApproxOMv}(G, B)$}{
    $w \gets \vec{0} \in \R^n$ \\
    $A \gets L$ \tcp{Vertices where we have not found an edge.}
    \tcp{Step I(a): Degree reduction on $A$.}
    \For{$t = 1, \dots, T_L := \frac{10n^2\log n}{D}$}{
        Pick $u \in L$, $b \in B$ uniformly at random. \\
        If $(u, b) \in E(G)$, set $w_u \gets 1$, $A \gets A \setminus \{u\}$.
    }
    \tcp{Step I(b): Degree reduction on $B$.}
    \For{$t = 1, \dots, T_R = \frac{10n^2\log n}{D_R}$}{
        Pick $a \in A$, $b \in B$ uniformly at random. \\
        \If{\label{line:if}$(a, b) \in E(G)$}{
            For all edges $(a', b) \in E(G)$ with $a' \in A$, set $w_{a'} \gets 1$, $A \gets A \setminus \{a'\}$. \\
            $B \gets B \setminus \{b\}$.
        }
    }
    \tcp{Step II: Peeling matchings.}
    For $a \in L \setminus A$, call $\cA.\textsc{Insert}((a, a'))$. \tcp{$a'$ are all distinct.}
    For $b \in R \setminus B$, call $\cA.\textsc{Insert}((b', b))$. \tcp{$b'$ are all distinct.}
    \While{\label{line:while} $|M[A, B]| \ge \eps n$, for $M \gets \cA.\textsc{Matching}()$}{
        Let $A_M$ be the set of vertices in $A$ matched by $M[A, B]$. \\
        For $a \in A_M$, set $w_a \gets 1$, $A \gets A \setminus \{a\}$, $\cA.\textsc{Insert}((a, a'))$. \\
    }
    Return $\cA$ to its original state by deleting edges. \\
    \Return $w$.
  }
\end{algorithm}
We prove the correctness of \cref{algo:omv}.
\begin{lemma}
\label{lemma:correct}
For a bipartite graph $G$ with adjacency matrix $M$, and $B \subseteq R$ with indicator vector $v$, $\textsc{ApproxOMv}(G, B)$ as in \cref{algo:omv} returns a vector $w$ that with high probability satisfies $d(Mv, w) \le 4D\eps n$.
\end{lemma}
\begin{proof}
Note that any vertex $u \in L$ with $\deg^{G[L,B]}(u) \ge D_L$, some edge $(u, b)$ is picked with probability at least $1-(1-D_L/n^2)^{T_L} \ge 1-n^{-10}$. Thus, after Step I(a) in \cref{algo:omv}, $\deg^{G[A,B]}(a) \le D_L$ for all $a \in A$. Similarly, for all $b \in B$ with $\deg^{G[A,B]}(b) \ge D_R$, some edge $(a, b)$ where $a \in A, b \in B$ is picked in Step I(b) with probability at least $1-(1-D_R/n^2)^{T_R} \ge 1-n^{-10}$. Thus, after Step I(b), $\deg^{G[A,B]}(x) \le D_R = D$ with high probability for all $x \in A \cup B$.

If the while loop in line \ref{line:while} does not hold, then by how algorithm $\cA$ is called and \cref{lemma:induced}, $\mu(G[A,B]) \le 4\eps n$. Thus $|E(G[A, B])| \le 4D\eps n$, and at most $4D\eps n$ vertices $a \in A$ have $(Mv)_a = 1$. Because \cref{algo:omv} only ever sets $w_a = 1$ when it finds an edge $(a, b)$, we conclude that $d(Mv, w) \le 4D\eps n$ with high probability.
\end{proof}

We analyze the runtime of \cref{algo:omv}.
\begin{lemma}
\label{lemma:omvruntime}
If $\cA$ has already undergone preprocessing, then \cref{algo:omv} runs in time $\O(n^2/\sqrt{D} + n\cT)$ with high probability, and leaves data structure $\cA$ in the same state that it was at the start.
\end{lemma}
\begin{proof}
The runtime of Step I(a) of \cref{algo:omv} is clearly $\O(n^2/D_L) = \O(n^2/\sqrt{D})$. As argued in \cref{lemma:correct}, $\deg^{G[A,B]}(a) \le D_L$ for all $a \in A$, so $|E(G[A,B])| \le D_L|B|$. Thus, with high probability, line \ref{line:if} occurs $\O(D_L|B| \cdot T_R/(|A||B|))$ times. Each time costs $O(|A|)$ time, for a total of $\O(D_L T_R) = \O(n^2/\sqrt{D})$. Because $\cA$ is a $(1-\eps)$-dynamic matching algorithm with amortized time $\cT$ (\cref{def:matchingalgo}), the edge insertions and deletions to $\cA$ cost $O(n\cT)$ time. The runtime cost of returning the matchings $M$ is $O(n)$, because the size of $A$ reduces by $1$ per edge in the returned matchings.
\end{proof}
Combining \cref{lemma:correct,lemma:omvruntime} establishes the if direction of \cref{thm:equivmatch}.
\begin{proof}[Proof of if direction of \cref{thm:equivmatch}]
Assume that the dynamic matching algorithm $\cA$ has amortized update time $\cT := n^{1-c}\eps^{-C}$.
Pass all updates to $M$ in a dynamic approximate $\OMv$ algorithm to $\cA$ as edge updates. The amortized cost of these steps is $\cT$. Because we only wish for a $(1-\gamma)$-approximate dynamic $\OMv$ solution, we can afford to wait $\gamma n$ steps between recomputing the solution.
By \cref{lemma:correct,lemma:omvruntime}, we can correctly return a solution the $(1-4D\eps)$-approximate dynamic $\OMv$ in time $\O(n^2/\sqrt{D} + n\cT)$. Set $\gamma = 2D\eps$, so that the output is $(1-\gamma)$-approximate. For $D = \eps^{-1/2}$, $\gamma = 2\eps{1/2}$, and the query time is $\O(n^2\eps^{1/4} + n\cT)$. Because $\cT = n^{1-c}\eps^{-C}$, choosing $\eps = n^{-\alpha}$ for sufficiently small $\alpha$ gives a subquadratic query time, completing the proof.
\end{proof}

\section{Dynamic Vertex Cover}
\label{sec:vertexcover}

\subsection{Multiplicative weights framework}
\label{subsec:mwuvertexcover}
We give a multiplicative weights framework for finding a minimum vertex cover. We consider the decision version which is a covering linear program. On a bipartite graph $G$, determine whether there is a vector $y \in \R_{\ge0}^V$ such that $\sum_{v \in V} y_v = 1$ and $y_u + y_v \ge c$ for all edges $(u, v) \in E$. For $\lambda > 0$ consider the potential $\Phi(y) := \sum_{(u,v) \in E} \exp(-\lambda(y_u+y_v))$. Let $y^{(0)} = 0$, so that $\Phi(y^{(0)}) \le n^2$.

Once again, we abstractly define an MWU oracle for vertex cover that allows the MWU algorithm to make progress.
\begin{definition}[MWU oracle for vertex cover]
\label{def:voracle}
Given bipartite $G = (V, E)$, $y \in \R_{\ge0}^V$, and $c, \lambda, \eps > 0$, we say that $\Delta$ is the output of a \emph{MWU oracle} for vertex cover if it satisfies the following:
\begin{enumerate}
    \item \label{item:v1} $\sum_{v \in V} \Delta_v = 1$.
    \item (Value) \label{item:v2} $\sum_{v \in V} \Delta_v \left(\sum_{u \in N(v)} \exp(-\lambda(y_u+y_v)) \right) \ge (1-\eps)c \cdot \Phi(y).$
    \item (Width) \label{item:v3} $\Delta_v \le \eps/(2\lambda)$ for all $v \in V$.
\end{enumerate}
\end{definition}
We show that calling this MWU oracle $T = \O(\eps^{-O(1)})$ times gives a fractional vertex cover with value at most $(1+\eps)/c$. We start by bounding the potential decrease in a single iteration.
\begin{lemma}
\label{lemma:voraclesub}
Let $\Delta$ be the output of the MWU oracle. Then $\Phi(y+\Delta) \le (1-(1-2\eps)\lambda c) \Phi(y)$.
\end{lemma}
\begin{proof}
We calculate that
\begin{align*}
\Phi(y+\Delta) &= \sum_{e = (u, v) \in E} \exp(-\lambda(\Delta_u+\Delta_v)) \exp(-\lambda(y_u+y_v)) \\
&\le \sum_{e = (u, v) \in E} (1-(1-\eps)\lambda(\Delta_u+\Delta_v))\exp(-\lambda(y_u+y_v)) \\
&= \Phi(y) - (1-\eps)\lambda \sum_{v \in V} \Delta_v \left(\sum_{u \in N(v)} \exp(-\lambda(y_u+y_v)) \right) \\
&\le \Phi(y) - (1-\eps)^2\lambda c \cdot \Phi(y) \le (1 - (1-2\eps)\lambda c)\Phi(y).
\end{align*}
Here, the first inequality uses that $\exp(1-x) \le 1-x+x^2 \le 1-(1-\eps)x$ for $x \le \eps$ and the width condition (\cref{item:v3}). The final line uses \cref{item:v2}.
\end{proof}
This lets us upper bound the value of $\Phi$ after $T$ iterations.
\begin{corollary}
\label{cor:voraclesub}
Let $y^{(0)} = 0$ and $y^{(i+1)} = y^{(i)} + \Delta^{(i)}$ where $\Delta^{(i)}$ is an MWU oracle for $y^{(i)}$. Then $\Phi(y^{(T)}) \le \exp(-(1-2\eps)T\lambda c)n^2$.
\end{corollary}
From this we can extract a nearly feasible vector for the covering linear program.
\begin{corollary}
\label{cor:voracle}
Let $y^{(0)} = 0$ and $y^{(i+1)} = y^{(i)} + \Delta^{(i)}$ where $\Delta^{(i)}$ is an MWU oracle for $y^{(i)}$. Then the vector $\bar{y} := \frac{1}{T} y^{(T)}$ satisfies $\sum_{e \in E} \bar{y}_e = 1$ and $y_u + y_v \ge (1-2\eps)c - \frac{2\log n}{\lambda T}$ for all edges $(u, v) \in E$.
\end{corollary}
\begin{proof}
We know that $\exp(-\lambda(y_u+y_v)) \le \Phi(y^{(T)}) \le \exp(-(1-2\eps)T\lambda c)n^2$ by \cref{cor:voraclesub}. The result follows by taking logarithms and dividing by $\lambda T$.
\end{proof}

\subsection{Offline dynamic vertex cover}
\label{subsec:offlinevertexcover}

We give a general framework for returning a solution satisfying the conditions of the MWU oracle in \cref{def:voracle}.
\begin{lemma}
\label{lemma:voracleabstract}
Let $\lambda = \eps/(10c)$ and $T = 20\eps^{-2}\log n$, and $1/\mu(G) \ge c \ge 1/(2\mu(G))$. There is an algorithm that implements an MWU oracle by making $\O(\eps^{-2})$ calls to the following oracle: for a subset $S \subseteq R$, estimate $\deg^{G[L, S]}(v)$ up to a $(1 \pm \eps/10)$-multiplicative factor for all $v \in L$.
\end{lemma}
\begin{proof}
Let $\tilde{s} \in \R_{>0}^V$ satisfy for all $v \in V$:
\[ \sum_{u \in N(v)} \exp(-\lambda(y_u+y_v)) \ge \tilde{s}_v \ge (1-\eps/2)\sum_{u \in N(v)} \exp(-\lambda(y_u+y_v)). \]
Let $U$ be the largest set satisfying: $|U|^{-1} \sum_{v \in U} \tilde{s}_v \ge (1-\eps)c \cdot \Phi(y).$ Such $U$ can be found in nearly-linear time by sorting by $\tilde{s}_v$, and taking the $k$ largest for some $k$. Set $\Delta_v = 1/|U|$ for $v \in U$.

We claim that this choice of $\Delta$ satisfies the conditions of the MWU oracle of \cref{def:voracle}. \cref{item:v1} follows trivially, and \cref{item:v2} follows by the definition of $\tilde{s}$ and $U$. To conclude, it suffices to establish that $|U| \ge 2\lambda/\eps = 1/(5c)$. Because $c \ge 1/(2\mu(G))$, it suffices to prove that $|U| \ge \mu(G)$. Indeed, if $U$ is a vertex cover of size $\mu(G)$ then 
\begin{align*}
\frac{1}{|U|} \sum_{v \in U} \tilde{s}_v &\ge \frac{1-\eps/2}{|U|} \sum_{v \in U} \sum_{u \in N(v)} \exp(-\lambda(y_u+y_v)) \\ &\ge (1-\eps/2) \sum_{e = (u, v) \in E} \frac{1}{\mu(G)} \exp(-\lambda(y_u+y_v)) \ge (1-\eps/2)c \cdot \Phi(y).
\end{align*}
By maximality of $U$, we know that $|U| \ge \mu(G)$ as desired.

Finally we argue that finding the estimates $\tilde{s}_v$ can be implemented with the oracle, which asks to estimate $\deg^{G[L,S]}(v)$ for a subset $S$. By \cref{item:v3} and $T = 20\eps^{-2}\log n$, we know that $0 \le \lambda y_v \le T\eps = O(\eps^{-1}\log n)$. Partition $R = R_0 \cup \dots \cup R_t$ where $R_i = \{v \in R: i\eps/10 \le \lambda y_v \le (i+1)\eps/10\}$. For any $v \in L$, we can estimate $\sum_{u \in N(v)} \exp(-\lambda(y_u+y_v)) = \exp(-\lambda y_v) \sum_{u \in N(v)} \exp(-\lambda y_u)$ by estimating $\deg^{G[L,R_i]}(v)$ up to $1\pm\eps/10$ for $i = 0, \dots, t$. We can do a symmetric procedure for $v \in L$. The number of calls is $2t \le \O(\eps^{-2})$.
\end{proof}
We start by giving a procedure to reduce the number of vertices in our dynamic graph to be proportional to the size of the minimum vertex cover, and how to recover a vertex cover.

Let us describe the approach at a high level. Let $G$ be a graph whose maximum vertex cover is size $s$, and let $A \cup B$ be a a constant-factor vertex cover of $G$. Let $H$ be a random \emph{contraction} (see \cref{def:contraction}) of $G$ onto $\Theta(s\eps^{-2})$ vertices. We maintain a vertex cover $S$ on $H$. To map $S$ back to $G$, do the following. First find all vertices in $A, B$ that map to something in $S$ -- this is the subset $\bar{A} \cup \bar{B}$ of $A \cup B$ that we will use in the vertex cover. Other than that, we are forced to take all neighbors of vertices in $A \setminus \bar{A}$ or $B \setminus \bar{B}$. This is a vertex cover because $A \cup B$ was a vertex cover in $G$. The challenging part is to establish that with high probability, $S$ is an approximately minimal vertex cover in $G$.
\begin{lemma}
\label{lemma:contractvtx}
Let $G = (V = L \cup R, E)$ be a bipartite graph with maximum vertex cover size $s$. Let $A \cup B$ be a vertex cover of $G$, with $|A| + |B| \le 10s$. Let $H = (V' = L' \cup R', E')$ be a random contraction of $G$ with $|L'| = |R'| = \Theta(s\eps^{-3})$, represented by map $\phi: V(G) \to V(H)$. Let $S$ be a $(1+\eps)$-approximate minimum vertex cover of $H$.

Let $A_1 := \{a \in A : |\phi^{-1}(\phi(a))| > 1\}$, $A_2 = \{a \in A \setminus A_1 : \phi(a) \in S$, and $\bar{A} = A_1 \cup A_2$. Define $\bar{B}$ similarly. Finally, let $\bar{S} = \bar{A} \cup \bar{B} \cup \bigcup_{a \in A \setminus \bar{A}} N(a) \cup \bigcup_{b \in B \setminus \bar{B}} N(b)$. Then $\bar{S}$ is a vertex cover of $G$ with $|S| \le (1+O(\eps))s$, with high probability.
\end{lemma}
We start with an abstract helper lemma.
\begin{lemma}
\label{lemma:helper}
Let $S_1, \dots, S_s \subseteq [n]$, and let $\phi: [n] \to [n']$ be a random map, for $n' = \Theta(s\eps^{-3})$. With high probability over $\phi$, for all subsets $I \subseteq [s]$:
\begin{align} \left|\bigcup_{i \in I} \phi(S_i)\right| \ge \min\left\{\left|\bigcup_{i \in I} S_i\right| - \eps s, s\right\}. \label{eq:want} \end{align}
\end{lemma}
To see the connection with \cref{lemma:contractvtx}, imagine that $S_i = N(i)$ for some $i \in A$. Then \cref{lemma:helper} says that with high probability over the random contraction $\phi$, that the size of the union of neighborhoods does not decrease by more than $\eps s$, unless the size was already larger than $s$ (which was the minimum vertex cover size already).

It is worth noting that to show \cref{lemma:helper}, na\"{i}vely applying a union bound over all $2^s$ subsets $I$ does not work. This is because for a fixed subset $I$, the failure probability is something like $\exp(-\eps s)$. Instead, we do the following.
\begin{proof}[Proof of \cref{lemma:helper}]
Let $J = [s]$. We will form a partition of $J = J_1 \cup \dots \cup J_t \cup J^*$, with $t \le 10/\eps$, as follows. If there is a subset $J' \subseteq J$ with $|J'| \ge \eps s/10$, $\left|\bigcup_{i \in J'} S_i\right| \le 10s$, then add $J_i = J'$ to the partition, and let $J \gets J \setminus J'$. Let $J^*$ be $J$ at the end of this procedure, and $J_0 = J_1 \cup \dots \cup J_t$.

Let $U_0 = \bigcup_{i \in J_0} S_i$. Let $\cE_0$ be the event that $|\phi(U_0)| \le |U_0| - \eps s/2.$ Note that $|U_0| \le 10st \le 100s/\eps$. Let $n' = Cs\eps^{-3}$, for a sufficiently large constant $C$, and $p = (100s/\eps)/n' \le c\eps^2$, for a small constant $c$. Now we bound the failure probability of $\cE_0$. Consider choosing the values $\phi(u)$ sequentially, for $u \in U_0$. The probability that $\phi(u)$ collides with a previously chosen value is at most $p$. Thus, the probability that $\cE_0$ fails is at most:
\[ \binom{|U_0|}{\eps s/2} p^{\eps s/2} \le (600p\eps^{-2})^{\eps s/2} \le 2^{-\eps s/2}, \]
for sufficiently small choice of $c$. Thus, assume that $\cE_0$ holds for the remainder of the proof.

For a subset $I \subseteq [s]$, let $U_I := \bigcup_{i \in I} S_i$. We break into two cases depending on whether $|U_I| \ge 10s$ of $|U_I| < 10s$. In the former case, we need to bound the probability that $|\phi(U_I)| \le s$. Again, consider choosing $\phi(u)$ for $u \in U_I$ sequentially. The probability that all $\phi(u)$ lie in some subset of $[n']$ of size $s$ is at most
\[ \binom{n'}{s}(s/n')^{10s} \le 4^{-s}. \]
Thus, we can union bound this failure over all subsets $I$ of $[s]$.

In the other case, $|U_I| < 10s$. Let $I^* = I \cap J^*$. We claim that the number of distinct $I^*$ for which this holds is bounded by $\sum_{k=0}^{\eps s/10} \binom{s}{k}$. Indeed, otherwise there would be some $|U_{I^*}| < 10s$ and $|I^*| \ge \eps s/10$, so our algorithm would have partitioned $J^*$ further. Let $U_1 := U_{I^*} \cap U_0$ and $U_2 = U_{I^*} \setminus U_1$. Because $\cE_0$ holds, we know that $|\phi(U_1)| \ge |U_1| - \eps s/2.$ Finally, we iterate over all $u \in U_2$, and bound the probability that more than $\eps s/2$ such $u$ have $\phi(u) = \phi(u')$ for some $u' \in U_0$, or previous $u' \in U_2$. For each $u$, the probability of a collision is at most $p$. Thus, the probability that more than $\eps s/2$ such $u$ collide is at most:
\[ \binom{10s}{\eps s/2} \cdot (2p)^{\eps s/2} \le (c\eps^{-1})^{\eps s/2}. \]
This suffices to union bound over the $\sum_{k=0}^{\eps s/10} \binom{s}{k}$ possible sets $I^*$. The lemma follows if both this event and $\cE_0$ hold, which we have argued holds with high probability.
\end{proof}
We are now ready to establish \cref{lemma:contractvtx}.
\begin{proof}[Proof of \cref{lemma:contractvtx}]
A vertex cover of $G$ can be chosen generically as follows. Choose $A' \subseteq A, B' \subseteq B$, and choose the vertex cover $A' \cup B' \cup \bigcup_{a \in A \setminus A'} N(a) \cup \bigcup_{b \in B \setminus B'} N(b)$. This is a vertex cover because $A \cup B$ is. By taking $S_i = N(i)$ for $i \in A$ (and same for $B$), \cref{lemma:helper} implies that in $G, H$ the vertex cover resulting from $A', B'$ have size off by $\eps s$, unless the total size is much larger than $s$.
\end{proof}
To give our offline algorithm, we consider the same setting as we did in \cref{subsec:mwuoracle}. We have graphs $G_1, \dots, G_t$, all with the same vertex set $L \cup R$, on $n' = \Theta(s\eps^{-3})$ vertices. Additionally, $\mu(G_i) \ge s$ for all $i \in [t]$. Finally, each $G_i$ differs from $G_1$ in at most $\Gamma$ edges. It is simple to use fast matrix multiplication to find degrees of vertices in all $G_i$.
\begin{lemma}
\label{lemma:degrees}
Let $G_1, \dots, G_t$ be bipartite graphs on $n' = \Theta(s\eps^{-3})$ vertices where each $G_i$ differs from $G_1$ in at most $\Gamma$ edges. There is an algorithm that give $S_1, \dots, S_t \subseteq R$, outputs $\deg^{G_i[L,S_i]}(v)$ for all $v \in L$, in total time $\O(\eps^{-O(1)} T(s, s, t) + t\Gamma)$.
\end{lemma}
\begin{proof}
Multiply the adjacency matrix of $G_1$ by the matrix of indicator vectors of $S_1, \dots, S_t$ to find $\deg^{G_1[L,S_i]}(v)$ for all $i = 1, \dots, t$ and $v \in L$. Then, update these to $\deg^{G_i[L,S_i]}(v)$ by updating $\Gamma$ values per $i \in [t]$. The total time is $\O(T(n', s, t) + t\Gamma)$.
\end{proof}
We are ready to give our offline algorithm for vertex cover.
\begin{proof}[Proof of \cref{thm:mainuppervtx}]
Let $s$ be an $O(1)$-approximation of the minimum vertex cover, and consider time steps when $\mu(G) \in [s, 10s]$. Let $H$ be a random contraction (\cref{def:contraction}) down to $\Theta(s\eps^{-3})$ vertices. Let us first estimate the amortized time needed to find a $(1+\eps)$-approximate minimum vertex cover on $H$. Let $H_i$ be the graph after $(i-1)\eps s$ edge updates, for $i = 1, \dots, t$, for $t = s^x$. Then, $\Gamma \le st = s^{1+x}$. By \cref{lemma:voracleabstract,lemma:degrees} we can find a $(1+\eps)$-approximate vertex cover on $H_1, \dots, H_t$ in total time:
\[ \O(\eps^{-O(1)}(T(s,s,s^x) + s^{1+2x})). \]
Now, we must map these vertex covers in $H_i$ back to $G_i$ via \cref{lemma:contractvtx}. This involves computing $\bigcup_{a \in A_i} N_{H_i}(a)$ for various sets $A_i$. Analogous to \cref{lemma:degrees}, this can be done in time:
\[ \O(T(n,n,s^x) + t\Gamma) = \O(T(n,n,s^x) + s^{1+2x}). \]
Note that we need to use a $n \times n$ matrix in the multiplication, because we must do the multiplication on the adjacency matrix of $G$ itself. Thus, the amortized time over $\eps st$ steps is bounded by:
\[ \eps^{-O(1)} \cdot \O(T(n,n,s^x) + s^{1+2x})/(\eps st) = \O(\eps^{-O(1)}(T(n,n,s^x)s^{-1-x} + s^x)). \]
\cite{GP13} gives an alternate way of computing an approximate vertex cover in amortized $\O(\eps^{-O(1)}s)$ time. The idea is to maintain an approximate vertex cover $A \cup B$, and at most $O(s)$ neighbors of each $a \in A$ or $b \in B$ (maintaining more neighbors cannot decrease the vertex cover below size $O(s)$ obviously). Thus, the total time is $O(s^2)$, which needs to be recomputed every $\eps s$ steps.

Thus, the algorithm of \cite{GP13} establishes \cref{thm:mainuppervtx} when $s < n^{.723}$. Otherwise, we obtain amortized runtime $\O(\eps^{-O(1)}(T(n,n,s^x)s^{-1-x} + s^x))$. For each $s \ge n^{.723}$, we can check that this runtime is $O(\eps^{-O(1)} n^{.723})$ for some choice of $x$, via \cite{Complexity}.
\end{proof}

\subsection{Equivalence of dynamic vertex cover and online matrix-vector}
\label{subsec:reducevertexcover}

\subsubsection{Vertex cover from $\OMv$}

We use the dynamic $\OMv$ algorithm of \cref{lemma:dynomv} to implement the oracle required by \cref{lemma:voracleabstract}.
\begin{lemma}
\label{lemma:implementvoracle}
Let $\cA$ be a dynamic $\OMv$ algorithm with amortized update time $n^{1-\delta}$ and query time $n^{2-\delta}$ against adaptive adversaries.
Then there is a randomized algorithm that supports updates to a dynamic bipartite graph $G$ with vertex set $L \cup R$, and queries of the following form: given a set $S$, compute $(1\pm\eps)$ multiplicative approximate estimates to $\deg^{G[L, S]}(v)$ for all $v \in L$. The algorithm succeeds with high probability with amortized update time $n^{1-\delta}$ and query time $\O(\eps^{-O(1)}n^{2-\delta/2})$.
\end{lemma}
\begin{proof}
To handle updates to $G$, pass the updates to the dynamic $\OMv$ algorithm $\cA$. This is amortized $n^{1-\delta}$ time by assumption. Now we implement the query.
Let $D = n^{\delta/2}$. For a vertex $v$, if $\deg^{G[L, S]}(v) \ge D$, then taking $\O(\eps^{-O(1)}n/D)$ random samples $s \sim S$ and checking whether $(v, s) \in E$ gives a $(1+\eps)$-approximate estimate of $\deg^{G[L, S]}(v)$. The total time of this step is $\O(\eps^{-O(1)}n^2/D)$.

It remains the estimate the degree of small-degree vertices, i.e., $v$ with $\deg^{G[L, S]}(v) \le D$. In this case, we compute the degree exactly. We use the same trick as in the proof of \cref{lemma:lowdegree,claim:degreeone}. Let $S_1, \dots, S_t$ for $t = \O(D)$ be random subsets of $S$, with $S_i \sim_{1/D} S$.
Then using the binary representation trick of \cref{claim:degreeone}, we can find all adjacent between $v$ and $S$ for all $v$ by calling $\O(D)$ Boolean matrix-vector multiplications (which are queries to $\cA$).
Thus, the total time of this step is amortized $\O(n^{2-\delta}D)$. For the choice $D = n^{\delta/2}$, both steps are time $\O(\eps^{-O(1)}n^{2-\delta/2})$ as desired.
\end{proof}
Now, we can establish the ``if'' direction of \cref{thm:equivvtx}.
\begin{proof}[Proof of if direction of \cref{thm:equivvtx}]
For each value of $s = 2^k$, consider the time steps when the minimum vertex cover size is in $[s, 10s]$ (this is doable by using a constant factor approximate dynamic vertex cover algorithm with $\O(1)$ update time).
If $s \le n^{2-\delta/4}$, then using the algorithm of \cite{GP13} gives amortized update time $\O(\eps^{-O(1)}s) \le \O(\eps^{-O(1)}n^{2-\delta/4})$. Otherwise, the algorithm can wait for $\eps s$ update before a recomputation of the vertex cover. Then, use \cref{cor:voracle}, \cref{lemma:voracleabstract}, and \cref{lemma:implementvoracle}, to compute a $(1+\eps)$-approximate fractional vertex cover in total time $\O(\eps^{-O(1)}n^{2-\delta/2})$ time per query. A fractional vertex cover can be rounded to integral in $O(n)$ time. Thus, the amortized update time is $\O(\eps^{-O(1)}n^{2-\delta/2})/(\eps s) \le \O(\eps^{-O(1)}n^{2-\delta/4})$, as desired.
\end{proof}
Combining this reduction (the if direction of \cref{thm:equivvtx}) with \cref{cor:dynomv} shows \cref{thm:onlineupperv}.

\subsubsection{$\OMv$ from vertex cover}
We start by recalling that $\OMv$ is equivalent to $\uMv$.
\begin{lemma}[\!\!{\cite[Lemma 2.11]{HKNS15}}]
\label{lemma:umv}
If the $\OMv$ conjecture is true, then there is no algorithm that solves the following in $n^{3-\delta}$ time. Given a bipartite graph $G = (V = L \cup R, E)$, and $n$ online queries of the form $A \subseteq L, B \subseteq R$, determine whether $G[A, B]$ has an edge.
\end{lemma}
Here, we have reinterpreted $\uMv$ as its graphical version: determine whether $G[A, B]$ is empty.
\begin{proof}[Proof of only if direction of \cref{thm:equivvtx}]
Assume there is a vertex cover data structure $\cA$ with amortized update time $\cT := O(n^{1-c}\eps^{-C})$. Consider a $\uMv$ instance with graph $G$. Initialize $\cA$ on $G$ by adding $O(n^2)$ edges. When given $A, B$, apply the reduction of \cref{lemma:inducedvertex} to find an additive $2\eps n$ vertex cover $S$ of $G[A, B]$ in $O(n)$ updates. If $|S| > 2\eps n$, we know that $G[A, B]$ has an edge. Otherwise, check all pairs of vertices in $A \times B$ with an endpoint in $S$ in time $O(\eps n^2)$. Thus, the total preprocessing and query time over $n$ queries is: $O(n^2 \cT + \eps n^3)$. Along with \cref{lemma:umv}, this contradicts the $\OMv$ conjecture when $\eps = n^{-\delta}$ for sufficiently small $\delta$.
\end{proof}

\section*{Acknowledgments}
We thank Arun Jambulapati, Rasmus Kyng, Richard Peng, and Aaron Sidford for helpful discussions, and Ryan Alweiss for discussion related to \cref{lemma:helper}. This work is partially supported by NSF DMS-1926686.

\bibliographystyle{alpha}
\bibliography{refs}

\end{document}